\documentclass[11pt]{article}

\usepackage[a4paper,left=15mm,top=20mm,right=15mm,bottom=20mm]{geometry}

\usepackage{graphicx} 
\usepackage[utf8]{inputenc}
\usepackage{amsmath, amsfonts, amssymb, amsthm}
\usepackage{enumerate}
\usepackage{xcolor}
\usepackage{ stmaryrd }
\usepackage{mathtools}
\usepackage{hyperref}
\usepackage{todonotes}

\usepackage{float}

\usepackage[linesnumbered,boxed,noend]{algorithm2e}

\newtheorem{theorem}{Theorem} 
\newtheorem{lemma}{Lemma}
\newtheorem{corollary}{Corollary}

\newtheorem{proposition}{Proposition}

\definecolor{byzantium}{rgb}{0.74, 0.2, 0.64}

\DeclareMathOperator{\lca}{lca}

\newcommand{\T}{\mathcal{T}}

\newcommand{\G}{\mathcal{G}}

\newcommand{\metricname}{Path-Label Reconciliation (PLR)}
\newcommand{\metricshortname}{PLR}
\newcommand{\dml}{d_{plr}}   

\usepackage{authblk}

\usepackage[authoryear]{natbib}
\setcitestyle{authoryear,open={(},close={)}}
\let\cite\citep

\providecommand{\keywords}[1]
{
  \small	
  \textbf{\textit{Keywords---}} #1
}

\begin{document}

\setlength{\marginparsep}{-0.6cm}
\setlength{\marginparwidth}{2.3cm}

\title{The Path-Label Reconciliation (PLR) Dissimilarity Measure for Gene Trees}

\author[1]{Alitzel {L{\'o}pez S{\'a}nchez}}
\author[2,3,4]{José Antonio {Ramírez-Rafael}}
\author[2]{Alejandro {Flores-Lamas}}
\author[2]{Maribel {Hernández-Rosales}}
\author[1]{Manuel Lafond}

\affil[1]{Department of Computer Science, University of Sherbrooke, J1K2R1 Quebec, Canada. \authorcr \texttt{\{manuel.lafond,alitzel.lopez.sanchez\}@usherbrooke.ca}}

\affil[2]{Center for Research and Advanced Studies of the National Polytechnic Institute, Irapuato, Gto., Mexico. \authorcr \texttt{\{jose.ramirezra, alejandrofloreslamas, maribel.hr\}@cinvestav.mx}}

\affil[3]{Department of Computer Science and Interdisciplinary Center for Bioinformatics, University of Leipzig, Germany}

\affil[4]{Max Planck Institute for  Mathematics in the Sciences, Leipzig, Saxony, Germany}

\date{\ }

\setcounter{Maxaffil}{0}
\renewcommand\Affilfont{\scriptsize}

\maketitle

\abstract{   
\textbf{Background:} 
In this study, we investigate the problem of comparing gene trees reconciled with the same species tree using a novel semi-metric, called the Path-Label Reconciliation (\metricshortname) dissimilarity measure. 
This approach not only quantifies differences in the topology of reconciled gene trees, but also considers discrepancies in predicted ancestral gene-species maps and speciation/duplication events,  offering a refinement of existing metrics such as Robinson-Foulds (RF) and their labeled extensions LRF and ELRF.  A tunable parameter $\alpha$ also allows users to adjust the balance between its species map and event labeling components.

\textbf{Our contributions:} 
We show that PLR can be computed in linear time and that it is a semi-metric.  
We also discuss the diameters of reconciled gene tree measures, which are important in practice for normalization, and provide initial bounds on \metricshortname, LRF, and ELRF. To validate \metricshortname, we simulate reconciliations and perform comparisons with LRF and ELRF. The results show that \metricshortname~provides a more evenly distributed range of distances, making it less susceptible to overestimating differences in the presence of small topological changes,
while at the same time being computationally efficient. 
Our findings suggest that the theoretical diameter is rarely reached in practice.  
The \metricshortname~measure advances phylogenetic reconciliation by combining theoretical rigor with practical applicability. Future research will refine its mathematical properties,  explore its performance on different tree types, and integrate it with existing bioinformatics tools for large-scale evolutionary analyses.  
The open source code is available at: \url{https://pypi.org/project/parle/}.
}

\bigskip
\noindent


\sloppy

\section{Introduction}

During evolution, it is well-known that genes can be duplicated, lost, and transferred, resulting in evolutionary scenarios that differ from the history of the species that contain them.  Gene trees can therefore be discordant with their species trees, and \emph{reconciliation} aims to infer the macro-evolutionary events that explain the discrepancies.  Several models have been proposed to achieve this task, allowing duplications and losses~\cite{goodman1979fitting,zhang1997mirkin,gorecki2006dls,bonizzoni2005reconciling,durand2005hybrid,vernot2008reconciliation,lafond2012optimal,hernandez2012event,geiss2020best}, horizontal gene transfer~\cite{gorecki2004reconciliation,doyon2010efficient,bansal2012efficient,kordi2016exact,jacox2017resolution,scornavacca2017fast,nojgaard2018time,kordi2019inferring,weiner2021improved,schaller2021indirect}, incomplete lineage sorting~\cite{zhang2011gene,stolzer2012inferring,rasmussen2012unified,wu2014most,chan2017inferring,li2021multilocus}, and others (see e.g.~\cite{delabre2018reconstructing,li2019simultaneous,hasic2019gene,boussau2020reconciling,anselmetti2021gene,santichaivekin2021empress,liu2021maximum}).  In addition, some of these models support segmental events that affect multiple genes at once~\cite{page2001vertebrate,burleigh2008locating,bansal2008multiple,paszek2017efficient,dondi2019reconciling}, and some approaches infer histories based on parsimony whereas others are probabilistic~\cite{arvestad2003bayesian,aakerborg2009simultaneous,larget2010bucky}.

This variety of reconciliation models and algorithms is accompanied by a large diversity of software and tools to reconcile gene trees with species trees (examples include NOTUNG~\cite{durand2005hybrid}, DLCoal~\cite{rasmussen2012unified},  RANGER-DTL~\cite{bansal2018ranger}, ecceTERA~\cite{jacox2016eccetera}, Jane~\cite{conow2010jane}).  
Most of these tools infer, for each ancestral gene tree node, the ancestral species to which the gene belonged to, as well as the event that affected the gene.  
It is, however, difficult to assess the quality of the reconciliations produced by these approaches, even with the availability of high quality software to simulate gene tree evolution (e.g. SimPhy~\cite{mallo2016simphy}, Asymmetry~\cite{schaller2022asymmetree}, aevol~\cite{batut2013silico}, ZOMBI~\cite{davin2020zombi}).  A standard benchmarking idea would be to simulate reconciled gene trees and to compare the inferred scenarios with the true simulated ones.  However, it is not straightforward to perform this comparison.
Indeed, reconciled gene trees exhibit three types of valuable information: the tree topology, the gene-species map, and the event labeling.  While there exist metrics to measure discrepancies for each of those three criteria individually, we are not aware of any established method to measure disagreements in all three simultaneously.  
There is a large body of literature on measuring topological differences between trees (e.g. \cite{Puigbo2007-hn}, \cite{Goloboff2017-mr}, \cite{Savage1983-ss},\cite{Makarenkov2000-ki},\cite{Munzner2003-ww}, \cite{Wagle2024}).
In terms of gene-species mapping discordance, the \emph{path distance} metric~\cite{path_dist} applies to gene trees with identical topologies but possibly different species maps, and quantifies how far the species of corresponding nodes are in the gene trees.  The metric was mainly introduced to obtain medians in the reconciliation spaces of gene trees. If the gene trees differ, though, the metric cannot be used.

Perhaps the most relevant metric to compare reconciled gene trees is the recent \emph{labeled Robinson-Foulds (RF) distance}, now called ELRF, which accounts for differences in topology and event labeling.  Given two gene trees, the distance is the minimum number of edge contractions, edge expansions, and node label substitutions required to transform one gene tree into the other~\cite{LRF}.  It is unknown whether this distance can be computed in polynomial time, the main difficulty being that edge operations must have the same label on both endpoints.  The authors then proposed a variant of this metric, called LRF, in which edge contractions/expansions are replaced with node insertions/deletions, which can be computed in linear time~\cite{briand2022linear}. 
Although these are perhaps the only approaches specifically tailored for gene tree comparison, their usage has some disadvantages.  First, these distances do not take gene-species maps into consideration.  Second, the metric suffers from the same well-known shortcomings as the RF distance, see~\cite{lin2011metric} for a discussion on this (for instance, a single misplaced leaf can increase the distance dramatically).  
Another subtle but yet important aspect is the topological uncertainty that can be present in gene trees.  In particular, when ancestral species undergo gene duplication episodes (see e.g.~\cite{gorecki2024unifying,paszek2017efficient}), the corresponding gene trees may contain large duplication subtrees.  In this case, there is too little phylogenetic signal to infer the topology of such duplication subtrees accurately.  However, most approaches penalize discrepancies in those local parts of the gene trees as in any other part, even though predicting different speciation patterns should be more heavily penalized than in duplication clusters.

In this work, we introduce a novel approach for comparing gene trees that considers all the aforementioned components that play a role in reconciliations: the species tree, the gene tree, the labeling of their internal nodes by species and events, as well as duplication clusters. This method effectively circumvents the shortcomings of the RF distance. Given two reconciled gene trees on the same set of genes, our dissimilarity measure establishes a correspondence between the gene tree nodes from both trees and applies a penalty if the matched nodes differ in species or event label. As we demonstrate, due to the constraints inherent in reconciliation models, this approach implicitly penalizes topological disagreements between the gene trees, except when the discordance is solely due to consecutive duplication rounds within the same species.

Our measure also has the advantage of being computable in linear time. We first explore some theoretical properties of our approach and show that it functions as a semi-metric in the space of reconciled gene trees. We demonstrate that if non-binary gene trees are considered, the measure does not necessarily satisfy the triangle inequality, although this remains an open question for binary trees.  We also provide initial results on the diameters of the \metricshortname, LRF, and ELRF measures, which are important in practice for normalization.

We then validate our approach through experiments involving simulated reconciliations on the same set of leaves and calculation of various measures. We show that, as can be expected from previous knowledge, RF, LRF, and ELRF tend to produce large distances overestimating tree differences, which can result from a rapid increase in the distance values when, for example, 
a single leaf is misplaced.
In contrast, our measure effectively captures small, average, and large distances between reconciliations.  Therefore, PLR is established as the first reconciliation measure with greater variability than RF variants, and sensitivity to differences in every component of  evolutionary scenarios. 

Note that due to space constraints, some of the proofs were replaced by a sketch of the main idea, and the full detailed arguments can be found in the Appendix.

\section{Preliminary notions}
A \emph{tree} is a connected acyclic graph.  Unless stated otherwise, all trees in this paper are rooted.
For a tree $T$, we denote by $r(T)$ the root of $T$, by $V(T)$ and $E(T)$ its set of nodes and edges, respectively, and by $L(T)$ its set of leaves.  A non-leaf node is called \emph{internal}. For $u, v \in V(T)$, we write $u \preceq_T v$ if $u$ is a \emph{descendant} of $v$, i.e., if $v$ is on the path between $r(T)$ and $u$ (we write $u \prec_T v$ if $u \neq v$).  Then $v$ is an \emph{ancestor} of $u$. 
 If $u \neq r(T)$, then the \emph{parent} $p_T(u)$ of $u$ if the ancestor $v$ of $u$ such that $uv \in E(T)$, and $u$ is a \emph{child} of $v$. A tree $T$ is \emph{binary} if each internal node has two children, and $T$ is a \emph{caterpillar} if all internal nodes have at most one child that is an internal node (that is, $T$ is a path with leaves attached to its nodes).
 
 For $X \subseteq V(T)$, we denote by $\lca_T(X)$ the \emph{lowest common ancestor} of all the nodes in $X$.  When $|X| = 2$, we may write $\lca_T(u, v)$ instead of $\lca_T(\{u, v\})$.  For $v \in V(T)$, we write $T(v)$ for the subtree of $T$ rooted at $v$.  Note that $L(T(v))$ is the set of leaves that descend from $v$, which we call the \emph{clade} of $v$.  As a shorthand, we may write $L_T(v)$ to denote the clade of $v$, or $L(v)$ if $T$ is understood. The \emph{distance} between two nodes $u, v$ in $T$ is denoted $dist_T(u, v)$, i.e., the length of the undirected path in $T$ between $u$ and $v$.

\subsection{Species trees and reconciled gene trees.}\label{sect:sTree-and-recon}
A \emph{species tree} $S$ is a tree which we assume to be binary.
A \emph{reconciled gene tree} (with $S$) is a tuple $\G = (G, S, \mu, l)$ where $G$ is a tree in which each internal node has at least two children (possibly more), $S$ is a species tree, $\mu : V(G) \rightarrow V(S)$ maps nodes of $G$ to species in $S$, and $l : V(G) \rightarrow \{dup, spec, extant\}$ is an event labeling.  We also have the following requirements: 
\begin{enumerate}
    \item 
    \emph{Leaves are from extant species:} for every leaf $v \in L(G)$, $\mu(v) \in L(S)$ and $l(v) = extant$.  Moreover, every internal node $w \in V(G) \setminus L(G)$ satisfies $l(w) \in \{dup, spec\}$;

    \item 
    \emph{Time-consistency:} for any two nodes $u, v \in V(G)$, $u \preceq_G v$ implies $\mu(u) \preceq_S \mu(v)$;

    \item 
    \emph{Speciations separate species:}
    for any node $v \in V(G)$ such that $l(v) = spec$, we have $\mu(v) \in V(S) \setminus L(S)$ and $v$ has exactly two children $v_1, v_2$.  
    
    Moreover, denoting by $s_1, s_2$ the two children of $\mu(v)$ in $S$, we have that $\mu(v_1) \preceq_S s_1$ and $\mu(v_2) \preceq_S s_2$, or $\mu(v_2) \preceq_S s_1$ and $\mu(v_1) \preceq_S s_2$.
    
\end{enumerate}

If $\mu$ satisfies $\mu(v) = lca_S( \{\mu(x) : x \in L(v) \})$ for every node $v \in V(G)$, then $\mu$ is called the \emph{lca-mapping}~\cite{gorecki2006dls,bonizzoni2005reconciling}.  In this map, all genes map to the lowest possible species according to the rules of reconciliation.  
These concepts are illustrated in Figure~\ref{fig:mainexample}, which presents two reconciled gene trees that use the lca-mapping (see caption).
Note that our reconciled gene trees are not restricted to the \emph{lca-mapping}.
However, it is known that if $l(v) = spec$, then $\mu(v)$ must indeed be the lowest common ancestor of all the species that appear in the genes below $v$.  However, the converse is not required to hold, that is, a duplication could be mapped to the lowest common ancestral species (or above).
\medskip

\noindent
\textbf{Isomorphism between reconciled gene trees.}
Two reconciled gene trees $\G_1 = (G_1, S, \mu_1, l_1)$ and $\G_2 = (G_2, S, \mu_2, l_2)$ are \emph{isomorphic} if they have the same sets of leaves, use the same species tree, have the same topology (i.e., they branch in identical ways), and their corresponding nodes map to the same species and have the same label.  If this holds, we write $\G_1 \simeq \G_2$.  Formally, $\G_1 \simeq \G_2$ if there exists a bijection $\phi : V(G_1) \rightarrow V(G_2)$ such that the following holds:
\begin{itemize}
    \item 
    $L(G_1) = L(G_2)$ and, for each leaf $x \in L(G_1)$, $\phi(x) = x$. 
    In other words, each leaf of $G_1$ is mapped to the same leaf in $G_2$;

    \item 
    $uv \in E(G_1)$ if and only $\phi(u) \phi(v) \in E(G_2)$;

    \item 
    for every node $v \in V(G_1)$, $\mu_1(v) = \mu_2( \phi(v) )$ and $l_1(v) = l_2(\phi(v))$.
 
\end{itemize}

\subsection{The \metricname~dissimilarity measure}

Let $\G_1 = (G_1, S, \mu_1, l_1)$ and $\G_2 = (G_2, S, \mu_2, l_2)$ be two reconciled gene trees.
We say that $\G_1$ and $\G_2$ are \emph{comparable} if: (1) they are reconciled with the same species tree $S$; (2) $L(G_1) = L(G_2)$; and (3) for each leaf $x \in L(G_1)$, $\mu_1(x) = \mu_2(x)$ (that is, extant genes map to the same species in both trees).
Unless stated otherwise, we assume that all pairs of reconciled trees mentioned are comparable, although (3) could be dropped, see remark below.

For a node $v \in V(G_1)$, we need a corresponding node for $v$ in $G_2$.  
This can be done in multiple ways, and here we assign this corresponding node as the lowest possible node of $G_2$ that is an ancestor of all the descendants of $v$.   
To put it more formally, define 
\[
m_{\G_1, \G_2}(v) = lca_{G_2}( L(G_1(v)) )
\]
which is the lowest common ancestor in $G_2$ of the clade of $v$.  Note that this is well-defined since $L(G_1) = L(G_2)$.  For instance in Figure~\ref{fig:mainexample}, $m_{\G_1, \G_2}(x_1) = y_0$.
When $\G_1, \G_2$ are clear from the context, we may write $m(v)$ instead of $m_{\G_1,\G_2}(v)$.
In essence, this is the lca-mapping, but applied between two gene trees.  Note that such mappings are usually applied between gene and species trees, but~\cite{kuitche2017reconstructing} also introduced the ancestral gene-gene map idea (or more specifically, ancestral RNA-gene maps).

Our measure has two components: one for the discrepancies in the species mappings, and one for the labelings.  These components are defined as:
\begin{align*}
    d_{path}(\G_1, \G_2) &= \sum_{v \in V(G_1)}  dist_S( \mu_1(v), \mu_2( m(v)) ) 
    \\
    d_{lbl}(\G_1, \G_2) &= | \{ v \in V(G_1) : l_1(v) \neq l_2(m(v)) \} |
\end{align*}
In words, in $d_{path}$, each term $dist_S(\mu_1(v), \mu_2(m(v)))$ penalizes $v$ by how far its species is from the species of its correspondent $m(v)$, and $d_{lbl}$ is simply the number of nodes of $G_1$ whose label differ from their correspondent in $G_2$.

\begin{figure}[t]
    \centering
    \includegraphics[width=\textwidth]{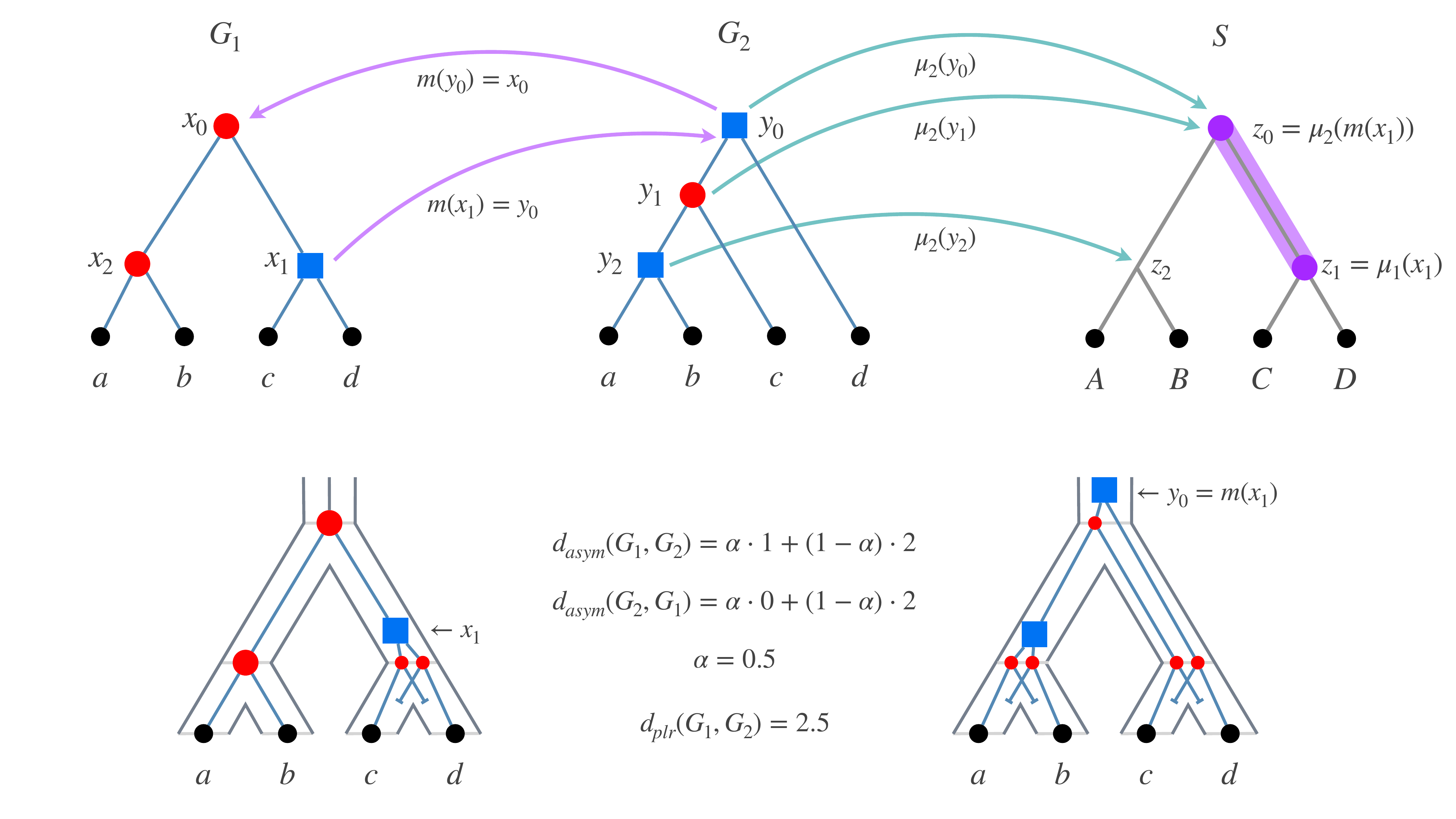}
    \caption{In the upper row, there are two reconciled gene trees $G_1$ and $G_2$ as well as a species tree $S$.  The event labelings are shown as red circles and blue squares, which represent speciations and duplications, respectively. Lowercase letters $a,b,c,d$ depict extant genes, while the corresponding uppercase letters are the species where genes reside. 
    The maps $\mu_1, \mu_2$ use the lca-mapping, that is, $\mu_1(x_0) = z_0, \mu_1(x_1) = z_1, \mu_1(x_2) = z_2$, and $\mu_2(y_0) = \mu_2(y_1) = z_0, \mu_2(y_2) = z_2$.
    The gene trees have the same set of leaves but different topology and event labeling.
    Purple arrows exemplify the maps $m_{\G_1, \G_2}(x_1)$, which is the lca of genes $c$ and $d$, and $m_{\G_2, \G_1}(y_0)$, while green arrows illustrate the species map $\mu_2$.
    The shaded edge in $S$ displays the path distance between $\mu_1(x_1) = z_1$ and $\mu_2(m(x_1)) = \mu_2(y_0) = z_0$.
    The lower row shows the explicit evolution of the gene trees within the species tree.
    The contribution of $x_1$ to the $d_{path}$ component is $1$, because $dist_S( \mu_1(x_1), \mu_2( m(x_1)) ) = 1$, whereas its contribution to $d_{lbl}$ is $0$ because $l(x_1) = l(m(x_1)) = dup$.  On the other hand, the node $y_0$ from $G_2$ contributes $0$ to $d_{path}$ since its correspondent $x_0$ is mapped to the same species, but contributes $1$ to $d_{lbl}$ since $l(y_0) = dup$ and $l(x_0) = spec$.
    }
    \label{fig:mainexample}
\end{figure}

We assume the existence of a given parameter $\alpha \in [0, 1]$ to weigh these components, and define the \emph{asymmetric dissimilarity} between $\G_1$ and $\G_2$ as: 
\[
d_{asym}(\G_1, \G_2) = \alpha \cdot d_{path}(\G_1, \G_2) + (1 - \alpha) \cdot d_{lbl}(\G_1, \G_2).
\]
Note that when $\alpha = 1$ and $G_1, G_2$ have the same topology, then $d_{asym}$ is exactly the \emph{path distance metric} studied in~\cite{path_dist}.  Our dissimilarity measure generalizes this by allowing trees with different topologies and by considering node labels.  One could ignore the $\alpha$ parameter by weighing $d_{path}$ and $d_{lbl}$ equally, which can be achieved with $\alpha = 0.5$.
Also notice that $d_{path}$ may be adapted to species trees with branch lengths.

It is easy to see that $d_{asym}$ is not symmetric.  For instance, suppose that $\G_1$ consists of a binary gene tree with several internal nodes mapping to different species,  and $\G_2$ consists of a star tree with a single internal node, such that both roots are duplications that map to the same species.  Then $d_{asym}(\G_1, \G_2)$ can be proportional to the number of internal nodes of $G_1$, whereas $d_{asym}(\G_2, \G_1) = 0$.  

The \metricname~dissimilarity is therefore defined as 
\[
\dml(\G_1, \G_2) = d_{asym}(\G_1, \G_2) + d_{asym}(\G_2, \G_1)
\]

If $\G_1$ and $\G_2$ are not comparable, then we define $\dml(\G_1, \G_2) = \infty$.

In Figure~\ref{fig:mainexample} we exemplify all the components of the dissimilarity measure.
In the example, following the $\mu_1, \mu_2$ maps given in the caption,  if we count the respective costs of $x_0, x_1, x_2$, we have $d_{path}(\G_1, \G_2) = 0 + 1 + 0 = 1$ and $d_{lbl}(\G_1, \G_2) = 1 + 0 + 1 = 2$.
If we put $\alpha = 0.5$, we get $d_{asym}(\G_1, \G_2) = 0.5 \cdot 1 + 0.5 \cdot 2 = 1.5$.
As for the costs of $y_0, y_1, y_2$, we get 
$d_{path}(\G_2, \G_1) = 0 + 0 + 0$ and $d_{lbl}(\G_2, \G_1) = 1 + 0 + 1 = 2$, and thus $d_{asym}(\G_2, \G_1) = 1$.  
Therefore, $\dml(\G_1, \G_2) = 2.5$.

\medskip

\noindent
\textbf{A remark on leaves belonging to the same species.}
Recall that condition (3) of comparability requires $\mu_1(x) = \mu_2(x)$ for every leaf $x \in L(G_1)$. 
Although this assumption usually follows from the knowledge of the species of a gene, it may not hold in some contexts.  Indeed, in metagenomics even the species of extant genes is unknown and needs to be inferred (see for example~\cite{gorecki2024unifying}). 
 Therefore, for an extant gene $x$, two different reconciliation algorithms may predict that $x$ belongs to a different species, leading to $\mu_1(x) \neq \mu_2(x)$.  Although condition (3) is useful in the proofs that follow, we note that it is not required in the definition of $\dml$, and the latter remains well-defined even if we drop this condition. 
Therefore, $\dml$ could be used to also compare gene trees with predicted gene-species maps that differ even at the level of leaves (although the theory developed hereafter may need revision for this case).

\medskip

\noindent
\textbf{A remark on setting $\alpha$.}\label{sect:alphaRemark}
The reader may notice that if $\alpha$ is ignored in $\dml$, or set to a constant, the $d_{path}$ component can easily outweigh the $d_{lbl}$ component. 
This is because in the worst case, $d_{path}(\G_1, \G_2)$ can be in $\Theta(nm)$, where $n$ is the number of species leaves and $m$ is the number of gene tree leaves, which occurs if most nodes of $\G_1$ are mapped to nodes of $\G_2$ with $\Theta(n)$ path distance in $S$ (see the diameter section for a detailed analysis).  On the other hand, the $d_{lbl}(\G_1, \G_2)$ component is always $O(m)$, as it only depends on the number of nodes in the gene tree.
This quadratic-versus-linear effect can be prevented by making $\alpha$ depend on $n$.  For instance, one may put $\alpha = 1/n$, or more generally $\alpha = c/n$ for some constant $c$.

\medskip

\noindent
\textbf{A remark on scenarios with horizontal transfer events.}\label{sect:transRemark}
In the presence of horizontal gene transfers, gene tree nodes can also undergo a \emph{transfer} event, and a different notion of time-consistency than ours is typically used (see e.g.~\cite{Njgaard2018}).  Nonetheless, such reconciliations also include a gene-species map $\mu$ and a labeling function $l$, and $\dml$ is also well-defined in this context.  On the other hand, it is unclear whether path distances are appropriate to compare transferred genes, and again, the theory that follows may need to be adapted to allow transfers.

\paragraph*{Least duplication-resolved gene trees.}
Consider a reconciled gene tree $\G = (G, S, \mu, l)$.  
If, in $G$, there is a connected subtree consisting only of duplication nodes, all mapped to the same species, then it is difficult to postulate on the exact topology of the duplication subtree due to the lack of clear phylogenetic signals.  
One solution is to contract the subtree into a single node to model the uncertainty. 
Contracting weakly supported branches in gene trees can be useful to detect and correct errors in dubious duplication nodes~\cite{lafond2013error}. Moreover, special cases of least-duplication resolved trees such as discriminating co-trees arise in the context of orthology detection \cite{Hellmuth2012,Gei2020}.
To this end, we say that an edge $uv \in E(G)$ is \emph{redundant} if $\mu(u) = \mu(v)$ and $l(u) = l(v) = dup$.  We then say that $\G$ is \emph{least duplication-resolved} if no edge $uv$ of $G$ is redundant.

Suppose that $\G$ is \emph{not} least duplication-resolved, and let $uv \in E(G)$ be a redundant edge, with $u = p_G(v)$.  We denote by $\G / uv$ the reconciled gene tree obtained by contracting $uv$ in $G$ and updating $\mu$ and $l$ accordingly.  More specifically, $\G / uv = (G', S, \mu', l')$, where: $G'$ is obtained from $G$ by deleting $v$ and its incident edges and, for each child $v'$ of $v$ in $G$, adding the edge $uv'$; and then putting $\mu'(w) = \mu(w)$ and $l'(w) = l(w)$ for every $w \in V(G')$.  If $R \subseteq E(G)$ is a set of redundant edges of $\G$, then $\G / R$ is the reconciled gene tree obtained after contracting every edge in $R$, in any order.
If $R$ is the set of all redundant edges of $\G$, then we define $LR(\G) = \G / R$, called the \emph{least duplication-resolved subtree} of $\G$.
It is not difficult to see that such a subtree is unique, least duplication-resolved, and satisfies all conditions of a reconciled gene tree.
Figure~\ref{fig:least-res} shows two gene trees and their least duplication-resolved version (note that two consecutive duplications in distinct species remain).

For two reconciled gene trees $\G_1, \G_2$, we write $\G_1 \simeq_d \G_2$ if $LR(\G_1) \simeq LR(\G_2)$.  This means that $\G_1$ and $\G_2$ may differ, but every form of disagreement is due to redundant edges, and they become identical in their least duplication-resolved form.  The following will be useful.

\begin{lemma}\label{lem:dup-lr-structure}
    Let $\G = (G, S, \mu, l)$ be a reconciled gene tree that is least duplication-resolved. Let $u,v \in V(G)$ be such that $v \prec_G u$.  
    Then either $\mu(u) \neq \mu(v)$ or $l(u) \neq l(v)$.
\end{lemma}


\begin{proof}
    Let $u = u_1, u_2, \ldots, u_k = v$ be the path from $u$ to $v$ in $G$.  
    Suppose that there is a speciation on the path, that is, there is some $i \in \{1, 2, \ldots, k - 1\}$ such that $l(u_i) = spec$.
    By the \emph{speciations separate species} requirement, denoting $s = \mu(u_i)$ and letting $s_1, s_2$ be the children of $s$ in $S$, we have $\mu(u_{i+1}) \preceq s_1$ or $\mu(u_{i+1}) \preceq s_2$.  Either way, $\mu(u_{i+1}) \prec \mu(u_i)$.  
    By the \emph{time-consistency} requirement, we then have 
    \[
    \mu(v) = \mu(u_k) \preceq \mu(u_{k-1}) \preceq \ldots \preceq \mu(u_{i+1}) \prec \mu(u_i) \preceq \mu(u_{i-1}) \preceq \ldots \preceq \mu(u_1) = \mu(u)
    \]  
    The presence of a $\prec$ in this chain implies $\mu(v) \neq \mu(u)$, as desired.

    So suppose that $l(u_1) = \ldots = l(u_{k-1}) = dup$.  
    If $l(v) \neq dup$, we are done, so assume $l(v) = l(u_k) = dup$.  By the definition of duplication least-resolved, we must have $\mu(u_k) \neq \mu(u_{k-1})$.
    By time-consistency, this implies $\mu(v) = \mu(u_k) \prec \mu(u_{k-1}) \preceq \mu(u_1) = \mu(u)$ and we are done.
\end{proof}

\section{Properties of the \metricname~dissimilarity}
We first show that in terms of time complexity, $\dml(\G_1, \G_2)$ can be computed in linear time, using appropriate data structures, in a very straightforward manner as shown in Algorithm~\ref{alg:dg1g2}.  
The details of a linear-time implementation can be found in the proof of Theorem~\ref{thm:linear-time}.

\begin{algorithm}[h]
\SetAlgoNoLine
\SetKwProg{Fn}{function}{}{}
\Fn{getAsymmetricDist($\G_1 = (G_1, S, \mu_1, l_1), \G_2 = (G_2, S, \mu_2, l_2)$, $\alpha$)}{

    $d_{path} \gets 0, d_{lbl} \gets 0$\;
    $m \gets lcamap(G_1, G_2)$\tcp*{Computes all $m(v) = lca_{G_2}( L(G_1(v)) )$}

	\ForEach{$v \in V(G_1)$}
	{
		$v' \gets m(v)$\;
		$d_{path} \gets d_{path} + dist_S(\mu_1(v), \mu_2(v'))$\;
		\lIf{$l_1(v) \neq l_2(v')$}{$d_{lbl} \gets d_{lbl} + 1$}
	}
	return $\alpha \cdot d_{path} + (1 - \alpha) \cdot d_{lbl}$\;
   
 }
\caption{Computing $d_{asym}$ in one direction.}
\label{alg:dg1g2}
\end{algorithm}

\begin{theorem}\label{thm:linear-time}
   The value $\dml(\G_1, \G_2)$ can be computed in time $O(|V(G_1)| + |V(G_2)| + |V(S)|)$.
\end{theorem}

\begin{proof}
   We argue that Algorithm~\ref{alg:dg1g2} can be implemented to run in time $O(|V(G_1)| + |V(G_2)| + |V(S)|)$, 
   which clearly proves the statement since we only need to run it twice (once for $\G_1$ versus $\G_2$, and once for $\G_2$ versus $\G_1$).
   We assume that $G_1$, $G_2$, and $S$ are pre-processed to answer lowest common ancestor queries between any two nodes in constant time.  
   This pre-processing time is linear for each tree~\cite{bender2000lca}, and therefore this step takes time $O(|V(G_1)| + |V(G_2)| + |V(S)|)$.
   We also assume that we know the depth of each node $x$ of $S$, denoted $depth(x)$, which is the distance between $x$ and the root.  This can easily be computed by a linear-time preorder traversal of $S$.
   It is not difficult to compute $m = lcamap(G_1, G_2)$ in time $O(|V(G_1)| + |V(G_2)|)$ using the $lca$ pre-processing and dynamic programming.  Indeed, for a gene tree node $v \in V(G_1)$ with children $v_1, \ldots, v_l$, we have $m(v) = lca_{G_2}( \{m(v_1), \ldots, m(v_l) \})$.  The latter $lca$ expression can be computed with $l - 1$ $lca$ queries as follows.   Define $w_{1, i} = lca_{G_2}( \{m(v_1), \ldots, m(v_i)\} )$.  First compute $w_{1,2} = lca_{G_2}(m(v_1), m(v_2))$, then $w_{1,3} = lca_{G_2}(w_{12}, m(v_3))$, and so on until $m(v) = w_{1, l} = lca_{G_2}(w_{1,l-1}, m(v_l))$, each in $O(1)$ time.  Since $l$ is the number of edges between $v$ and its children, the number of $lca$ queries required throughout the execution of the whole algorithm is less than the number of edges of $G_1$, which is $O(|V(G_1)|)$.
   
   For each $v \in V(G_1)$, we can obtain $dist_S(\mu_1(v), \mu_2(v'))$ in constant time, since it is equal to $depth(\mu_1(v)) + depth(\mu_2(v')) - 2 \cdot depth( lca_S(\mu_1(v), \mu_2(v') ) )$.
   It follows that each $v \in V(G_1)$ can be dealt with in $O(1)$ time and the loop of the algorithm takes time $O(|V(G_1)|)$, which does not add to the complexity.
\end{proof}

\subsection*{A semi-metric under least duplication-resolved equivalence}
Let us recall the mathematical notion of a metric, which can be defined as a triple $(X, d, \equiv)$ where $X$ is a set, $d : X \times X \rightarrow \mathbb{R}$ is a dissimilarity function, and $\equiv$ is a binary equality operator on $X$, such that the following conditions are satisfied:
\begin{itemize}
    \item 
    (identity) for all $x \in X$, $d(x, x) = 0$;

    \item 
    (positivity) for all $x, y \in X$, if $x \not \equiv y$, then $d(x, y) > 0$;

    \item 
    (symmetry) for all $x, y \in X$, $d(x, y) = d(y, x)$;
    
    \item 
    (triangle inequality)
    for all $x, y, z \in X$, $d(x, z) \leq d(x, y) + d(y, z)$.
\end{itemize}

If all the above conditions are satisfied, except the triangle inequality, then $(X, d, \equiv)$ is a \emph{semi-metric}.
If $X$ is clear from the context, we call $\emph{d}$ a \emph{metric (or semi-metric) under $\equiv$}.

In our case, we consider the set of all reconciled gene trees, with $\dml$ as our dissimilarity function.  As for the equality operator, we may consider $\simeq$ or $\simeq_d$.
In general, $\dml$ does not always meet the \emph{positivity} requirement under $\simeq$.  That is, $\G_1 \not\simeq \G_2$ does not necessarily imply $\dml(\G_1, \G_2) > 0$.
Consider for example two gene trees with different topologies, but whose internal nodes are all duplications in the same species (in which case all internal nodes incur a path and label penalty of $0$).  For a more elaborate example, see Figure~\ref{fig:least-res}.

\begin{figure}
    \centering
    \includegraphics[width=0.85\textwidth]{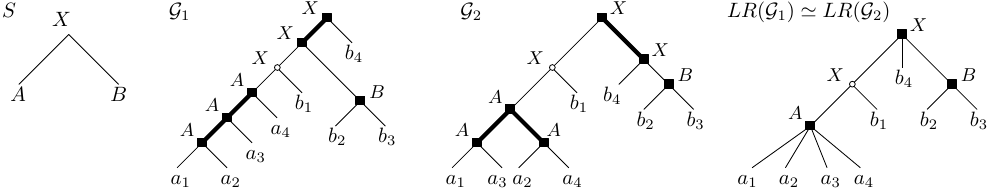}
    \caption{Two different reconciled gene trees $\G_1, \G_2$, where redundant edges are bold (again, lowercase letters indicate the species).  Their $\dml$ value is $0$ (one can check that all duplications in species $W \in \{A, X, B\}$ in either tree maps to a duplication in the same $W$ in the other tree, and the $X$ speciation to an $X$ speciation.  On the right, the least duplication-resolved version of the trees, showing that $\G_1 \simeq_d \G_2$.}
    \label{fig:least-res}
\end{figure}

However, we can show that $\dml$ is a semi-metric under $\simeq_d$.  The most difficult part is to show that $\G_1 \not \simeq_d \G_2$ implies $\dml(\G_1, \G_2) > 0$.  We first need to show that contracting the trees to their least duplication-resolved form cannot increase the dissimilarity.

\begin{lemma}\label{lem:ldr-cost}
    Let $\G_1 = (G_1, S, \mu_1, l_1), \G_2 = (G_2, S, \mu_2, l_2)$ be comparable reconciled gene trees, and let $uv \in E(G_1)$ be a redundant edge.  Then $\dml(\G_1, \G_2) \geq \dml(\G_1 / uv, \G_2)$.
\end{lemma}


\begin{proof}
    Denote $\G_1 / uv = \G'_1 = (G'_1, S, \mu'_1, l'_1)$, and note that $V(G'_1) = V(G_1) \setminus \{v\}$.  Also observe that contractions do not alter the set of descendants of a node, and thus $L_{G_1}(x) = L_{G'_1}(x)$ for all $x \in V(G'_1)$.
    Therefore, for $w \in V(G_1) \setminus \{v\}$, we have $m_{\G_1, \G_2}(w) = m_{\G'_1, \G_2}(w)$.  Since $w$ and its correspondent $m_{\G_1, \G_2}(w)$ both have the same species map and label before and after the contraction, the contribution of $w$ to $d_{path}$ and $d_{lbl}$ is the same in either $\G_1$ and $\G'_1$.  As this holds for every $w$ that is still in $G'_1$, we get $d_{asym}(\G_1, \G_2) \geq d_{asym}(\G'_1, \G_2)$.

    Now let $w \in V(G_2)$.  If $m_{\G_2, \G_1}(w) \neq v$, then $m_{\G_2, \G'_1}(w) = m_{\G_2, \G_1}(w)$, since $m_{\G_2, \G_1}(w)$ is still a common ancestor of $L(w)$ in $G'_1$, and no such lower ancestor can exist as it would also exist in $G_1$.
    The contribution of $w$ to $d_{path}$ and $d_{lbl}$ is therefore unchanged.
    Suppose instead that $m_{\G_2, \G_1}(w) = v$.  Then in $G'_1$, $u$ is a common ancestor of $L_{G_2}(w)$, and no such lower ancestor could exist, as it would also be in $G_1$.  In other words, $m_{\G_2, \G'_1}(w) = u$.  Since $uv$ is redundant, $\mu'_1(u) = \mu_1(u) = \mu_1(v)$ and $l'_1(u) = l_1(u) = l_1(v)$.  As the contribution of $w$ to $d_{path}$ and $d_{lbl}$ is based on $\mu_1(v)$ and $l_1(v)$, it is unchanged in $\G'_1$, and so $d_{asym}(\G_2, \G_1) = d_{asym}(\G_2, \G'_1)$, which concludes the proof.
\end{proof}

Since Lemma~\ref{lem:ldr-cost} can be applied to any sequence of contractions, in either $\G_1$ or $\G_2$ by symmetry, we get the following.

\begin{corollary}\label{cor:ldr-cost}
    Let $\G_1 = (G_1, S, \mu_1, l_1), \G_2 = (G_2, S, \mu_2, l_2)$ be reconciled gene trees with the same leafset.  Then $\dml(\G_1, \G_2) \geq \dml(LR(\G_1), LR(\G_2))$.
\end{corollary}

The above is sufficient to deduce that if $\simeq_d$ is interpreted as ``being the same reconciled tree'', then we have a semi-metric, unless $\alpha = 0$ or $\alpha = 1$ (see Appendix for full proof).

\begin{theorem}
    \label{thm:semimetric}
    For any $\alpha \in (0, 1)$, $\dml$ is a semi-metric under $\simeq_d$. 
\end{theorem}


\begin{proof}
    \emph{Identity.}  Let $\G = (G, S, \mu, l)$ be a reconciled gene tree.   Let us argue that $d(\G, \G) = 0$. 
    Let $v \in V(G)$, and notice that $m_{\G, \G}(v) = v$.  Therefore, the distance in $S$ between $\mu(v)$ and $\mu(m(v))$ is $0$ and $v$ incurs no label penalty.  Since this holds for every $v$, the dissimilarity between $\G$ and $\G$ is $0$. 

    \emph{Symmetry.}  Observe that $\dml$ is symmetric by design, as it adds both terms $d_{asym}(\G_1, \G_2)$ and $d_{asym}(\G_2, \G_1)$ whether we calculate $\dml(\G_1, \G_2)$ or $\dml(\G_2, \G_1)$.

    \emph{Positivity.}
    The rest of the proof is dedicated to showing that the positivity requirement is met under $\simeq_d$.
    Let $\G_1 = (G_1, S, \mu_1, l_1)$ and $\G_2 = (G_2, S, \mu_2, l_2)$  such that $\G_1 \not\simeq_d \G_2$.  We need to show that $\dml(\G_1, \G_2) > 0$.
    We may assume that $\G_1, \G_2$ are least duplication-resolved.  This is because if not, then by Corollary~\ref{cor:ldr-cost}, $\dml(\G_1, \G_2) \geq \dml(LR(\G_1), LR(\G_2))$, so if we prove positivity for any pair of least duplication-resolved gene trees,  it will also hold for any pair of trees.
    Moreover, $\G_1 \not \simeq_d \G_2$ means that their least duplication-resolved forms are not isomorphic.
    
    So, from now on we assume that $\G_1$ and $\G_2$ are least duplication-resolved, and that $\G_1 \not \simeq \G_2$.
    To ease notation, we use $m_{12}$ instead of $m_{\G_1, \G_2}$ and $m_{21}$ instead of $m_{\G_2, \G_1}$.  To ease further, for $v \in V(G_1)$, we may denote $v' = m_{12}(v)$ for the correspondent of $v$ in $G_2$.

    Suppose first that for every $v \in V(G_1)$, $L(v) = L(m_{12}(v))$ \emph{and} that for every $w \in V(G_2)$, $L(w) = L(m_{21}(w))$.  This means that both trees have exactly the same set of clades.
    We claim that there must be some $v \in V(G_1)$ such that $\mu_1(v) \neq \mu_2(v')$ or $l_1(v) \neq l_2(v')$.  If this is true, then $\dml(\G_1, \G_2) > 0$ as desired.  For the sake of contradiction, assume otherwise that $\mu_1(v) = \mu_2(v')$ and $l_1(v) = l_2(v')$ for every $v \in V(G_1)$.  We claim that $m_{12}$ is an isomorphism and deduce that $\G_1 \simeq \G_2$, which will contradict our assumption.  
    
    First note that $m_{12}$ is a bijection.  Indeed, no two distinct nodes $v_1, v_2$ of $G_1$ have $L(v_1) = L(v_2)$, because nodes with a single child are forbidden.  Thus each $v \in V(G_1)$ maps to a unique and distinct node of $G_2$, namely the unique node $v'$ with $L(v) = L(v')$.  Hence $m_{12}$ is injective.  The map is also surjective: if there is some $w \in V(G_2)$ such that no $v \in V(G_1)$ maps to $w$, then the clade $L(w)$ is not present in $G_1$.  In that case, $L(m_{21}(w)) \neq L(w)$, contrary to our initial assumption.  It follows that $m_{12}$ is bijective.
    
    Next observe that $L(G_1) = L(G_2)$ and that $m_{12}$ maps leaves of $G_1$ to the same leaf in $G_2$, as required by the definition of $\simeq$.  Moreover by assumption, for every $v \in V(G_1)$, we have $\mu_1(v) = \mu_2(v') = \mu_2( m_{12}(v) )$ and $l_1(v) = l_2(v') = l_2( m_{12}(v) )$.  
    
    It only remains to argue that $m_{12}$ preserves edges and non-edges.
    Consider $uv \in E(G_1)$, with $u$ the parent of $v$.  
    Then $L(v) \subset L(u)$ and there is no node $w \in V(G_1)$ such that $L(v) \subset L(w) \subset L(u)$.
    In $G_2$, by clade equality we also have $L(v') \subset L(u')$, and so $v'$ must descend from $u'$.  If $u' v' \notin E(G_2)$, then the child $z$ of $u'$ on the path from $u'$ to $v'$ in $G_2$ satisfies $L(v') \subset L(z) \subset L(u')$.  By clade equality, this implies $L(v) \subset L(z) \subset L(u)$.  But then the clade of $m_{21}(z)$, the correspondent of $z$ in $G_1$, cannot be equal to the clade of $z$, since we argued that $G_1$ contains no clade sandwiched between $L(u)$ and $L(v)$.  Therefore, $u' v' \in E(G_2)$.  
    Using a symmetric argument, if $u'v' \in E(G_2)$, then $uv \in E(G_1)$, since again $u'$ and $v'$ must have corresponding clades in $G_1$ with none in-between.  
    Thus, $m_{12}$ satisfies all the conditions of an isomorphism, which is a contradiction as we assumed that $\G_1$ and $\G_2$ were not isomorphic.

    We may thus assume that there is some $v \in V(G_1)$ such that $\mu_1(v) \neq \mu_2(v')$ or $l_1(v) \neq l_2(v')$.  Either way, since $\alpha \in (0, 1)$ (i.e. $\alpha \neq 0, 1$), we get $\dml(\G_1, \G_2) > 0$.
    This takes care of the case where $G_1$ and $G_2$ have the same set of clades.
    
    So, we may assume that there is some $v \in V(G_1)$ such that $L(v) \neq L(m_{12}(v))$, or some $w \in V(G_2)$ such that $L(w) \neq L(m_{21}(w))$.    
    We may assume that the former occurs --- which is without loss of generality since we can swap the roles of $\G_1$ and $\G_2$ as $\dml$ is symmetric.

    Let $v \in V(G_1)$ such that $L(v) \neq L(v')$.  
    Because $v'$ is the lowest common ancestor of $L(v)$ in $G_2$, it must be that $L(v)$ is a strict subset of $L(v')$. 
    If $v$ and $v'$ have different species map or label, the dissimilarity will be non-zero and we are done, so assume that $\mu_1(v) = \mu_2(v')$ and $l_1(v) = l_2(v')$.  

    Let $v'' = m_{21}(v')$ be the node of $G_1$ that corresponds to $v'$.  
    We have $L(v) \subset L(v') \subseteq L(v'')$, and so $v''$ must be a strict ancestor of $v$.  Since $\G_1$ and $\G_2$ are least duplication-resolved, by Lemma~\ref{lem:dup-lr-structure}, $v''$ either has a different species or a different label than $v$, and thus different from $v'$ as well.  Since $v'$ has a difference with its correspondent $v''$, the dissimilarity is non-zero.
    Having handled every case, it follows that $\G_1$ and $\G_2$ have non-zero dissimilarity.
\end{proof}

We next show that, despite being a semi-metric, the $\dml$ dissimilarity measure is not a metric since it does not satisfy the triangle inequality on non-binary gene trees, regardless of $\alpha$.  If $\alpha$ is a constant, it can even be far from satisfying the inequality.

\begin{figure}[H]
    \centering
    \includegraphics[width=0.8\textwidth]{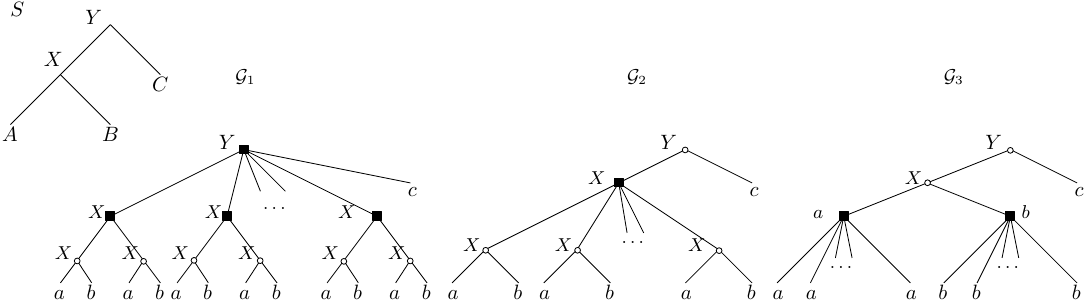}
    \caption{A species tree $S$ and reconciled gene trees $\G_1, \G_2, \G_3$ that violate the triangle inequality. }
    \label{fig:reconc-triangle}
\end{figure}

\begin{proposition}
    For any $\alpha \in [0, 1]$, possibly depending on the number of leaves of the gene trees, $\dml$~does not necessarily satisfy the triangle inequality under $\simeq_d$.  
    This is true even if the gene trees use the lca-mapping.

    Moreover, for any fixed $\alpha < 1$, the quantity $\dml(\G_1, \G_3)$ can be arbitrarily larger than $\dml(\G_1, \G_2) + \dml(\G_2, \G_3)$. 
\end{proposition}


\begin{proof}
    Consider the tree reconciled gene trees illustrated in Figure~\ref{fig:reconc-triangle}.
    Suppose that the root of the $\G_1$ gene tree has, as children, $k \geq 2$ copies of an $((a, b), (a, b))$ subtree.  Then the gene tree of $\G_2$ has $2k$ copies of an $(a, b)$ subtree.  Assume that every $(a, b)$ subtree of $\G_1$ maps uniquely to some $(a, b)$ subtree of $\G_2$.  
    One can check that the three reconciled gene trees use the lca-mapping.

    We now consider the asymmetric dissimilarity between each pair of reconciled gene trees.
    \begin{itemize}
        \item 
        $d_{asym}(\G_1, \G_2) = 1 - \alpha$.  The speciations in $X$ of $\G_1$ each map to a speciation in $X$ in $\G_2$, so they incur no cost.  Each duplication in $X$ of $\G_1$ must map to the duplication in $X$ of $\G_2$, again at no cost.  Only the duplication mapped to $Y$ maps to a speciation mapped to $Y$, incurring a cost of $1 - \alpha$.

        \item 
        $d_{asym}(\G_2, \G_1) = 1$.  The root of $\G_2$ is a speciation in $Y$ that maps to a duplication in $Y$ (cost $(1 - \alpha)$ for the label change) and the duplication in $X$ maps to the root of $\G_1$ (mapped to $Y$, cost of $\alpha$ for the path of distance $1$).  The $(a, b)$ subtrees incur no cost.

        \item 
        $d_{asym}(\G_2, \G_3) = 1 - \alpha$.  Each $(a, b)$ subtree in $\G_2$ maps to the speciation in $X$ of $\G_3$, so no cost incurred.  The only cost incurred is $1 - \alpha$ for the duplication in $X$ of $\G_2$ mapped to a speciation in $X$.

        \item 
        $d_{asym}(\G_3, \G_2) = 1 + \alpha$.  The duplications in $a$ and $b$ map to a duplication in $X$ (cost $2\alpha$) and the speciation in $X$ to a duplication in $X$ (cost $1 - \alpha$).
    \end{itemize}

    From the above, it follows that $\dml(\G_1, \G_2) + \dml(\G_2, \G_3) = 4 - \alpha$.
    
    Let us now compare $\G_1$ and $\G_3$:
    \begin{itemize}
        \item 
        $d_{asym}(\G_1, \G_3) = k(1 - \alpha) + 1 - \alpha$.  The $(a, b)$ subtrees of $\G_1$ incur no cost.   However, each duplication in $X$ maps to the speciation in $X$ of $\G_3$, and there are $k$ of them, for a cost of $k(1 - \alpha)$.  The $Y$ duplication also pays $(1 - \alpha)$.  
        \item 
        $d_{asym}(\G_3, \G_1) = 2 + 3\alpha$.  Every internal node of $\G_3$ maps to the root of $\G_1$.  The $a$ and $b$ duplications of $\G_3$ each map to a duplication in $Y$ of $\G_1$ for a cost of $2 \cdot 2\alpha$.  The speciation in $X$ costs $\alpha + 1 - \alpha = 1$ and the speciation in $Y$ costs $1 - \alpha$.  
    \end{itemize}
    By summing the two directions and simplifying, we get $\dml(\G_1, \G_3) = k(1 - \alpha) + 2\alpha + 3$.
    Observe that if $\alpha < 1$ is a fixed constant, the quantity $k(1 - \alpha) + 2\alpha + 3$ can be made arbitrarily large as $k$ grows, whereas $\dml(\G_1, \G_2) + \dml(\G_2, \G_3) = 4 - \alpha$ is a constant.  This justifies the second part of the proposition.

    If $\alpha$ is not necessarily fixed, then assuming $k \geq 2$ and $\alpha \leq 1$, 
    \begin{align*}
        k(1 - \alpha) + 2\alpha + 3 \geq 2 \cdot (1 - \alpha) + 2 \alpha + 3 = 5 > 4 - \alpha
    \end{align*}
    which establishes the lack of a triangle inequality.
\end{proof}

We observe that in the example from Figure~\ref{fig:reconc-triangle}, the triangle inequality is violated mainly because the trees are heavily imbalanced in terms of number of internal nodes.  We could not find counter-examples in which all trees are \emph{binary}.

\section{Diameters}

We now study the question of computing the \emph{diameter} of $\dml$, which is the maximum possible dissimilarity achievable over a given instance size.  This can be useful in practice for normalization, since we can compare heterogeneous datasets by dividing obtained dissimilarities by the diameter.  In the context of general trees, the diameter is usually the maximum dissimilarity among all pairs of trees with the same number of leaves $n$.  In reconciled gene trees though, there are multiple ways to define the diameter.  We may fix two numbers $n, m$, and find the maximum $\dml$ value among all species trees on $n$ leaves and pairs of gene trees on $m$ leaves.  Or, we could decide to fix the species tree $S$, and find the gene trees over $m$ leaves of maximum $\dml$ value with respect to $S$.  Or, we could fix the species tree $S$, and for each species leaf $s \in L(S)$ also fix the number $m_s$ of extant genes that belong to $s$, and find the most distant gene trees under these parameters.  

Even the simplest forms of diameters are not trivial to determine.  We thus provide initial results by determining the diameter in the case that the species tree $S$ is fixed, and gene trees contain exactly one gene per species. Even though this assumption may not hold in practice,
we hope that the bounds established here can be extended to broader classes of scenarios in the future.
 We leave the question of finding the theoretical values of the other diameters as open problems.

For a fixed species tree $S$, let $\mathsf{G}^S$ represent the set of all reconciled gene trees $\G = (G, S, \mu, l)$, such that for each $s \in L(S)$, exactly one leaf $x$ of $G$ satisfies $\mu(x) = s$.  Since each leaf of $G$ is uniquely identifiable by its species, we assume that all the elements of $\mathsf{G}^S$ have the same leaves and are pairwise-comparable.
We define the \emph{diameter for fixed $S$} as:
    \begin{equation*}
        diam(\dml,S) = \max\limits_{\G_1,\G_2 \in \mathsf{G}^S} \Big\{ \dml(\G_1,\G_2) \Big\}
        \end{equation*}

 In terms of $d_{lbl}$, in the worst case $d_{lbl}(\G_1, \G_2)$ is the number of internal nodes of the gene tree of $\G_1$, which occurs when all labels differ.
 We next characterize the maximum possible path distance.  
 It is tempting to make every node of $\G_1$ map to a deepest leaf of $S$, and every node of $\G_2$ to the root of $S$, thereby maximizing $dist_S(\mu_1(v), \mu_2(m(v)))$ for every node $v$, but such an example may not satisfy the rules of reconciliation. 

 For a species tree $S$, 
 let $H(S) = \sum\limits_{v \in V(S)\setminus L(S)} dist_S(v,r(S))$ be the sum of root-to-internal node distances.

 \begin{lemma}\label{lem:hs}
    Let $S$ be a species tree on $n \geq 1$ leaves. 
    Let $\G_1$ and $\G_2$ be two reconciled trees in $\mathsf{G}^S$. Then $d_{path}(\G_1,\G_2) \leq H(S) \leq (n - 1)(n - 2)/2$.
    \label{lem:general_path}
\end{lemma}

\begin{proof}
    Let us focus on $d_{path}(\G_1, \G_2) \leq H(S)$.
    Let $\G_1 = (G_1, S, \mu_1, l_1)$ and $\G_2 = (G_2, S, \mu_2, l_2)$.
    For the duration of the proof, let $\lambda : V(G_1) \rightarrow V(S)$ be the lca-mapping between $G_1$ and $S$.
Recall that $\lambda(v) = lca_S( \{ \mu_1(l) : l \in L(v) \} )$ is the lowest common ancestor of all the species that appear below $v$.  Also recall that $\mu_1(v) \succeq_S \lambda(v)$ for every $v \in V(G_1)$, since by time-consistency, $\mu_1(v)$ must be an ancestor of $\mu_1(l)$ for every $l \in L(v)$.

Let $v \in V(G_1) \setminus L(G_1)$ and denote $v' = m_{\G_1, \G_2}(v)$.  
We claim that 
\[dist_S( \mu_1(v), \mu_2(v') ) \leq dist_S( \lambda(v), r(S) ).
\]  
As mentioned, we have that $\mu_1(v) \succeq_S \lambda(v)$.  Moreover, by the definition of $m_{\G_1, \G_2}$, $L(v')$ contains all the leaves in $L(v)$, so we also deduce that $\mu_2(v')$ is an ancestor of $\mu_2(l) = \mu_1(l)$ for every leaf $l \in L_{G_1}(v)$ (since we assume that $\G_1$ and $\G_2$ are comparable).   
 Therefore, $\mu_2(v') \succeq_S \lambda(v)$ as well.
Then, $dist_S( \mu_1(v), \mu_2(v') )$ is a distance between two ancestors of $\lambda(v)$ in $S$, which is maximized when one node maps to $\lambda(v)$ and the other maps to $r(S)$ (noting that $v$ and $v'$ can take either of these two roles).  

We deduce that 
\[
d_{path}(\G_1, \G_2) = \sum_{v \in V(G_1) \setminus L(G_1)} dist_S(\mu_1(v), \mu_2(m(v))) \leq \sum_{v \in V(G_1) \setminus L(G_1)} dist_S( \lambda(v), r(S) )
\]
where we note that we do not need to sum over the leaves of $G_1$, since corresponding leaves are mapped to the same species and do not contribute to the path distance. 

    We will prove by induction on the number of leaves $n$ of $S$ that, for any gene tree $\G_1$ reconciled with $S$ that has one gene per species, the value of $\sum_{v \in V(G_1) \setminus L(G_1)} dist_S( \lambda(v), r(S) )$ is always at most $H(S)$. 
    For the base case, suppose that $S$ has $n=1$ leaf. In this case, note that the species tree $S$ consists of a single node $v$ and $H(S) = 0$.  Because $\G_1$ has a single leaf per species, $G_1$ also has a single leaf, and the summation evaluates to $0$.  As another base case, suppose that $n = 2$.  Then $S$ has a root $X$ and two leaves $S_1, S_2$, and $H(S) = 0$ again.  Then $\G_1$ must also consist of a root $r$ and two leaves mapped to $S_1, S_2$.  Moreover, $r$ must map to $r(S)$ by time-consistency, and $dist_S( \mu_1(r), r(S)) = 0$.
    
    For the induction step, suppose that $S$ has $n\geq 3$ leaves.  Let $\G_1 = (G_1, S, \mu_1, l_1)$ be a reconciled trees with $S$ that has one gene per species.
    
    Let $S_1, S_2$ be two leaves of $S$ with common parent $X$ (i.e., $(S_1, S_2)$ form a so-called \emph{cherry} of $S$).  Note that because $S$ is binary and because there exist other leaves in $S$, $X$ cannot be the root and it therefore has a parent $p(X)$ in $S$.   
    We denote by $S' := S - S_1$ the tree that results from removing $S_1$ and suppressing the resulting node of degree $2$.  It other words, to obtain $S - S_1$, remove node $S_1$ and its parent $X$ as well as their incident edges, add the edge $(p(X),S_2)$. 

    Let $s_1$ be the unique leaf in $G_1$ such that $\mu_1(s_1) =  S_1$, and denote by $t := p(s_1)$ its parent in $G_1$.  
    Let $G_1' := G_1 - s_1$ be the tree obtained by removing $s_1$ and its incident edge.  
    If $t$ becomes a non-root node of degree $2$, then add an edge between $p(t)$ and all children of $t$, then delete $t$ and its incident edges.  If instead $t$ becomes a root with a single child, then delete $t$ and its incident edge.
    Note that because gene trees may not be binary, it is possible that $t$ is still present in $G_1'$.

    Let $\lambda'$ be the lca-mapping between $G_1'$ and $S'$ obtained from $\mu_1$, which is well-defined since, for each $l \in L(G_1')$, $\mu_1(l) \in L(S')$.
    Since $S'$ has one less leaf that $S$, we may use induction and deduce that
    \[
    \sum_{v \in V(G'_1) \setminus L(G'_1)} dist_{S'}(\lambda'(v), r(S')) \leq H(S')
    \]
    We note that by the choice of $X$, we have $H(S) = H(S') + dist_S(X, r(S))$.
    
    To relate this quantity to $G_1$, consider a node $v \in V(G_1) \setminus L(G_1)$.
    Suppose that $v \neq t$, in which case $v$ is also in $G_1'$.  Note that $\lambda'(v)$ is a node of $S'$, but also of $S$.  
    By observing that $L_{G_1'}(v) \subseteq L_{G_1}(v)$, we infer that $\lambda(v) \succeq_S \lambda'(v)$, since $\lambda(v)$ must be an ancestor of all the leaves in $L_{G_1'}(v)$, plus possibly $s_1$, which can only ``raise'' the lca.  
    We claim that $\lambda'(v) \neq S_2$.  To see this, note that in $G'_1$, $v$ has at least two descending leaves $v_1$ and $v_2$.  Because $G_1$ and $G'_1$ have one gene per species, $v_1$ and $v_2$ belong to two distinct species of $S'$.  One of those is possibly $S_2$, but the other is not, which implies that $\lambda'(v) \succeq_{S'} lca_{S'}( \mu_1(v_1), \mu_2(v_2))$ cannot be $S_2$.  
    By this claim, in $S$ the internal node $X$ is not on the path between $\lambda'(v)$ and the root, and so the distance between $\lambda'(v)$ and the root is the same in either $S$ or $S'$.  The fact that $\lambda(v) \succeq_S \lambda'(v)$ then implies that $\lambda(v)$ can only be closer to the root of $S$.  
    For $v \neq t$, we thus have 
    \[
    dist_S(\lambda(v), r(S)) \leq dist_S(\lambda'(v), r(S')).
    \]
    Next, consider the node $t \in V(G_1)$, which may or may not be in $G_1'$.  
    Since $t$ is an ancestor of $s_1$, and of some other leaf $l$ belonging to some species other than $S_1$, we get that $\mu_1(t) \neq S_1$ and thus that $\mu_1(t)$ is a strict ancestor of $S_1$.  Under this condition, $dist_S(\lambda(t), r(S))$ is maximized when $\lambda(t) = X$.  

    Combining the facts gathered so far, we deduce that 
    \begin{align*}
        \sum_{v \in V(G_1) \setminus L(G_1)} dist_S(\lambda(v), r(S)) &= \sum_{v \in V(G_1) \setminus (L(G_1) \cup \{t\})} dist_S(\lambda(v), r(S)) + dist_S(\lambda(t), r(S)) \\
        &\leq \sum_{v \in V(G_1) \setminus (L(G_1) \cup \{t\})} dist_{S'}(\lambda'(v), r(S')) + dist_S(X, r(S)) \\
        &\leq H(S') + dist_S(X, r(S)) \\
        &\leq H(S)
    \end{align*}
    as desired.  

    Finally, we show that $H(S) \leq (n-1)(n - 2)/2$ by induction on $n$.
If $n = 1$ or $n = 2$, the longest path from $r(S)$ to an internal node has length 0, which verifies the base case.
So suppose $n \geq 3$.  Let $S_1$ be a deepest leaf of $S$, which must be part of a cherry formed by the leaf pair $S_1, S_2$ whose common parent is $X$.  By induction,  $H(S - S_1) \leq (n - 2)(n - 3)/2$.  If we add back $X$ to $S - S_1$, the lengths of the previous root-to-internal node paths is unchanged, and we only add a path of length at most $n - 2$ from $r(S)$ to $X$ (in the worst case, that path goes through all the $n - 1$ internal nodes of $S$).  We get $H(S) \leq (n - 2)(n - 3)/2 + n - 2 = (n - 1)(n - 2)/2$.  
\end{proof}

We can now proceed to prove the following theorem.

\begin{theorem}
    Let $S$ be a species tree on $n \geq 2$ leaves. Then
        \begin{equation*}
        diam(\dml,S) =  2\alpha \cdot H(S) + (1 - \alpha)(2n - 2).
    \end{equation*}
    Moreover, among all species trees with $n$ leaves, the diameter is maximized when $S$ is a caterpillar, in which case $diam(\dml, S) = \alpha (n - 1)(n - 2) + (1 - \alpha)(2n - 2)$.
\end{theorem}

\begin{proof}
   We first show that our expression is an upper bound for the diameter. 
 Let $\G_1 = (G_1, S, \mu_1, l_1)$ and $\G_2 = (G_2, S, \mu_2, l_2)$ be in $\mathsf{G}^S$. 
 For the label component of $\dml$, since $G_1$ and $G_2$ have the same number of leaves, the maximum number of different nodes is bounded by the maximum number of internal nodes per tree. Given that every species has exactly one gene, this is exactly $n-1$.  Hence, for the $d_{lbl}$ component, the cost is at most $2(n - 1) = 2n - 2$  considering both directions.
 As for the $d_{path}$ component, we know by Lemma~\ref{lem:hs} that the cost is at most $H(S)$ in each of the two directions.  This justifies the upper bound. 

 For an example that achieves this bound, suppose that $G_1$ and $G_2$ have the same topology as $S$ (that is, they are both a copy of $S$, but we replace each leaf by a gene from that species).  For $\mu_1$, we use the lca-mapping between $G_1$ and $S$, and put $l_1(v) = spec$ for every $v \in V(G_1) \setminus L(G_1)$ (which is possible since $G_1$ is a copy of $S$ and uses the lca-mapping).  
 For $G_2$, for every $v \in V(G_2) \setminus L(G_2)$, we put $\mu_2(v) = r(S)$ and $l_2(v) = dup$.  

 Note that because all internal node labels differ, this example achieves the maximum $d_{lbl}$ value possible.
 For the path component, let $s \in V(S) \setminus L(S)$.  Let $v \in V(G_1)$ be the corresponding node in $G_1$ and $v'$ the corresponding node in $G_2$ (i.e., the copy of $s$ in the trees).  
 Note that $m(v) = v'$ and $m(v') = v$, and that $\mu_1(v) = s, \mu_2(v) = r(S)$.  Hence, the contribution of $v$ and $v'$ to the path component is 
 $dist_S(s, r(S))$ on one side, plus $dist_S(r(S), s)$ on the other side.  
 Since this holds for every internal $s$ of $V(S)$, the cost of the $d_{path}$ component is $2H(S)$.

 As for the second part of the statement, fist notice that by Lemma~\ref{lem:hs}, $H(S)$ is never more than $(n-1)(n-2)/2$, and so $diam(\dml, S) \leq \alpha (n-1)(n-2) + (1-\alpha)(2n - 2)$. 
 Suppose that $S$ is a caterpillar.
 Notice that the deepest internal node $X$ satisfies $d(X, r(S)) = n - 2$ (there are $n - 1$ internal nodes in $S$, and the path goes through all of them).  Then $p(X)$ has a path of length $n - 3$ to $r(S)$, then $p(p(X))$ of length $n - 4$, and so on, so that $H(S) = \sum_{i=1}^{n-2} i = (n - 1)(n - 2)/2$, achieving the maximum possible $H(S)$.
\end{proof}

\paragraph*{On the labeled RF distances}

We now take a brief detour into another distance designed to compare reconciliations, namely the labeled Robinson-Foulds distances as presented in \cite{LRF,briand2022linear}, of which there are two variants.  These distances are used in the next section and we briefly discuss upper bounds on their diameters.
An edge of a tree is \emph{internal} if none of its endpoints is a leaf.
 \emph{labeled tree} is a pair $\T = (T, l)$ where $T$ is an unrooted tree without degree two nodes, and $l : V(T) \setminus L(T) \rightarrow X$ assigns some label from some set $X$ to each internal node (one can think of the label set as $X = \{spec, dup\}$).  A \emph{label-flip} is an operation that changes the label of an internal node.  An \emph{extension} is the reverse of a contraction: it takes a node $v$ and a non-empty subset $X$ of its neighbors, creates a new node $w$, deletes the edges $\{vx : x \in X\}$, then adds the edges $\{wx : x \in X\}$ along with $vw$, such that the latter must be internal.
A \emph{labeled contraction} is an operation that contracts an internal edge $uv$ satisfying $l(u) = l(v)$, and a \emph{labeled extension} is an extension of $v$ that creates node $w$ with $l(w) = l(v)$.  

Given two labeled trees $\T_1 = (T_1, l_1), \T_2 = (T_2, l_2)$, the \emph{ELRF distance}~\cite{LRF} between $\T_1$ and $\T_2$ is the minimum number of labeled contractions, labeled extensions, and label-flips required to transform $\T_1$ into $\T_2$.

The \emph{LRF distance}~\cite{briand2022linear} is the minimum number of contractions, extensions, and label-flips required to transform $\T_1$ into $\T_2$ (note that the authors use the notion of node deletions and insertions, but are stated in~\cite{briand2022linear} to be the same as contractions and extensions).

For an integer $n \geq 3$, the diameter of the ELFR (resp. LRF) distance is the largest possible distance among all possible labeled trees with $n$ leaves.  These diameters were not discussed in the literature.  We provide bounds which we believe to be tight, under the assumption that the label set consists of two elements $X = \{spec, dup\}$.

\begin{figure}[H]
    \centering
    \includegraphics[width=0.7\textwidth]{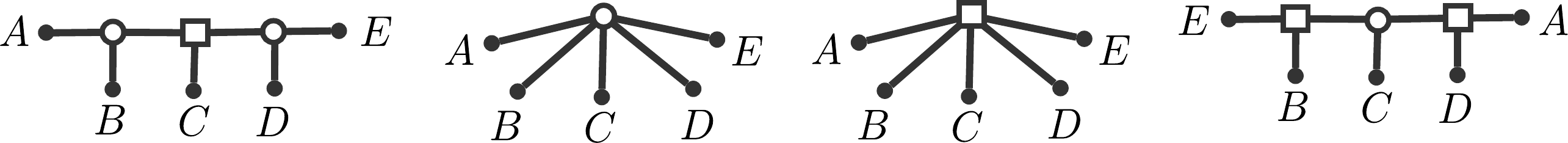}
    \caption{ An example of two labeled trees (left and right), with $n = 5$ leaves and two internal edges, which both need to be contracted.  To achieve this under the ELRF distance, we can perform $\lfloor (n - 2)/2 \rfloor = 1$ relabeling to make every label a circle (not shown), then contract every internal edge to obtain a star tree (second drawing).  We can then change the remaining label, and reverse the operations to obtain the right tree.  This takes $7 = 3n - 8$ operations.}
    \label{fig:diam}
\end{figure}

\begin{proposition}
    For $n \geq 3$ and label set $X$ of size $2$, the ELRF diameter is at most $3n - 8$.  Furthermore, the LRF diameter is at most $2n - 5$.  
\end{proposition}
\begin{proof}
    Let $\T_1 = (T_1, l_1), \T_2 = (T_2, l_2)$ be two labeled trees with $n$ leaves. 
    Consider the ELRF distance first.  
    Note that an unrooted tree has at most $n - 3$ internal edges (if we start with $n=3$ leaves, there are $0$ such edges, and by adding one leaf at a time, we create at most one internal edge for each leaf added).  
    We can make at most $n - 3$  contractions on $T_1$ to turn it into a star tree, i.e., a tree with a single internal node.  
    However, contractions must have the same label on both endpoints, so we may need to make all the labels identical before doing so.  There are at most $n - 2$ internal nodes in an unrooted tree. 
    
    If $n$ is odd, we can always achieve the same label everywhere with $\lfloor (n-2)/2 \rfloor = (n - 3)/2$ label flips by changing the label that occurs the least frequently.
    Thus, with $(n - 3)/2 + n - 3$ operations, we can transform $\T_1$ into a star tree.  We can do the same with $\T_2$, so one way of turning $\T_1$ into $\T_2$ is to make $\T_1$ a star tree, possibly flip the label of the internal node, then reverse the path from $\T_2$ to a star tree.  Counting each step, this results in at most 
    \[
    (n-3)/2 + n - 3 + 1 + n - 3 + (n-3)/2 = 3n - 8
    \]
    operations.  
    If $n$ is even, we have two cases.  If there are $(n - 2)/2$ of each label in both $\T_1$ and $\T_2$, we see that $(n - 2)/2 + n - 3$ operations suffice to turn $\T_1$ into the star tree with either label, and the same holds for $\T_2$, and thus that $2((n-2)/2 + n - 3) = 3n - 8$ operations suffice.  If, say, one label $x \in X$ occurs strictly more than $(n - 2)/2$ times in, say, $\T_2$, we can turn $\T_1$ into a star tree with $(n - 2)/2 + n - 3$ flips, flips the single label into $x$ if needed, then perform $n - 3$ extensions, and then strictly less than $(n-2)/2$ relabels.  This also results in at most $3n - 8$ operations.

    The LRF bound uses a similar idea, except that contractions do not need to have their endpoint labels identical.  We can thus perform at most $n - 3$ contractions on $\T_1$ to obtain a star tree, possibly flip the label of the internal node, and perform at most $n - 3$ extensions, adding the correct label each time, to obtain $\T_2$, resulting in $n - 3 + 1 + n - 3 = 2n - 5$ operations.
\end{proof}

The intuition is that we can always contract all $n - 3$ internal edges of the first tree.  In ELRF, we may have to relabel half of the $n - 2$ internal nodes to do this, so using $n - 3 + (n - 2)/2$ operations to reach a star tree (in the proof we show that this bound can be achieved while also attaining any desired label at the root of the star, with some case handling required for odd versus even $n$).  This has to be reversed, leading to $3n - 8$. In LRF, we can just contract all $n - 3$ internal edges directly, possibly relabel the internal node of the star tree, then extend.
It is possible that these bounds are tight.  Consider Figure~\ref{fig:diam} for the ELRF distance.  If we generalize this pattern, it would appear that we need to flip $\lfloor (n-2)/2 \rfloor$ nodes, do $n - 3$ mandatory contractions, flip the central node, and reverse the process.  This results in the upper bound $3n - 8$.  For LRF, one can think of a pair of trees with no label in common, that require the maximum of $2n-6$ contractions and extensions, plus a label flip.  However, proving that such examples cannot be handled better is not trivial, and since these distances are not the focus of the paper, we reserve those for future work.

\section{Methods}
We compared the distribution of the \metricshortname~semi-metric against the classical Robinson-Foulds (RF) and its ELRF and LRF variants. To this end, we designed and implemented a stepwise procedure to simulate reconciled trees.
The software tool to compute $\dml$ is available as open source at: \url{https://pypi.org/project/parle/}.

\subsection{Simulation of reconciliations}\label{sect:recSim}

The existing programs for simulation of reconciliations like AsymmeTree or SaGePhy \cite{AsymmeTree,Kundu-Bansal-2019} operate in a top-bottom fashion by first simulating ancestral genes/species followed by a birth-death process generating speciation, duplication, and loss events among others.
This procedure does not guarantee trees with a fixed set of genes, whereas 
the \metricshortname, LRF, and ELRF metrics require trees with the same set of leaves.
To fulfill this requirement, we designed Algorithm \ref{alg:randomRec}, which takes as input a species tree $S$, as well as a set of genes $\Gamma$ and the assignment of species $\sigma:\Gamma\rightarrow L(S)$, then builds a reconciled gene tree over leafset $\Gamma$ in a bottom-up fashion. At each iteration it picks pairs of genes $x',x''\in \Gamma$ and substitutes them with a newly created node $x$, being the parent of the chosen genes. Finally, $x$ is associated with an event and mapped to the species tree in Line~\ref{line:event}.
Algorithm \ref{alg:randomRec} uses the lca-mapping $\mu$ for the generated gene trees.  It is known that this map satisfies time-consistency, and that a node $x$ with children $x', x''$ can be a speciation if and only if $\mu(x) \notin \{\mu(x'), \mu(x'')\}${\cite{gorecki2006dls}}.  If this is not satisfied, the algorithm assigns $l(x) = dup$, and otherwise chooses $l(x) \in \{dup, spec\}$, which guarantees the \emph{speciations separate species} condition.

Algorithm \ref{alg:randomRec} considers a probability distribution $P$ of picking $x',x''\in \Gamma$. In our implementation, this probability decays exponentially w.r.t. the distance between the species where $x'$ and $x''$ reside, in other words, the larger $d=dist_S(\mu(x'), \mu(x''))$ is, the smaller the chance of choosing $x',x''$. In particular, we use the probability $e^{-0.7 d}$.
This approach is intended to prevent close elements in the gene tree from being mapped to distant nodes in the species tree, such a setting causes most of the inner nodes in the gene tree to be mapped near the root of the species tree, which would in turn create many $dup$ nodes.

In total, we generated 9 sets of random reconciliations, obtained as follows.
First, we generated three species trees $S_n$, where $n$ is the number of leaves: $S_{10}$, $S_{25}$, and $S_{50}$, using the \texttt{AsymmeTree} package \cite{AsymmeTree} under the \emph{innovations model} as described in \cite{innovations}.
For each species tree $S_i$ we generated the gene sets $\Gamma_{i,5}$, $\Gamma_{i,10}$, and $\Gamma_{i,20}$, together with the assignments of species $\sigma_{i,5}$, $\sigma_{i,10}$, and $\sigma_{i,20}$ in such a way that for the set $\Gamma_{i,j}$ each species $y \in L(S_i)$ contains at least one gene and at most $j$ genes. Considering this restriction, the number of genes for each species was chosen with uniform probability.\\

\begin{algorithm}[H]
\DontPrintSemicolon
\SetKwProg{Fn}{function}{}{}
\Fn{generate\_random\_scenario($S,\Gamma,\sigma$)}{
    \tcp{$S$ is a species tree, $\Gamma$ is the set of genes, $\sigma$ is a map from $\Gamma$ to their species.}
    Initialise $\G = (G,S,\mu,l)$ with $L(G) = \Gamma$ such that $\mu$ maps every leaf to their corresponding species in $L(S)$ according to $\sigma$\;
    \While{$|\Gamma|>1$}
    {
    Pick two genes $x',x''$ in $\Gamma$ according to a probability distribution $P$.\;
    Create a new node $x$ as the parent of $x'$ and $x''$.\;
    \tcp{Set reconciliation map and label of the new node.}

    Set $\mu(x) = lca_S(\mu(x'), \mu(x''))$\;
    \lIf {$\mu(x) \in \{\mu(x'), \mu(x'')\}$\label{line:event}}
     {$l(x) = dup$}
     \lElse{ choose $l(x)$ from $\{ dup, spec \}$ with uniform probability}
     $\Gamma\gets (\Gamma\setminus\{x',x''\})\cup\{x\}$\tcp*{Update set of genes}
    }
    return $\G$\;
  }
  
\caption{Simulation of random reconciliation scenarios.}
\label{alg:randomRec}
  
\end{algorithm}

\paragraph*{Distance distribution and normalization}\label{sect:distaces}
Given a set $R_{i,j}$ of random reconciliations generated from $S_i$ and $\Gamma_{i,j}$, we computed the \metricshortname, ELRF, LRF, and RF measures for each pair of different reconciliations. We set $|R_{i,j}|=50$, resulting in 1225 total comparisons per set of random reconciliations.
As argued in Section~\ref{sect:alphaRemark}, the parameter $\alpha$ of \metricshortname~is aimed to balance the quadratic-versus-linear components of the distance.
Following this analysis, we set $\alpha=1/i$ for the dataset $R_{i,j}$. Furthermore, to address the impact of $\alpha$ on the metric we also used the values 0, 0.25, 0.5, 0.75, and 1.

We normalized the distances obtained to have a fair comparison between the distribution of the different measures. We used two strategies, first, we normalized PLR by the theoretical diameter of the distance, while ELFR by its upper bond, and second by the empirical max normalization, which consists of dividing each computed value of a measure by the maximum encountered in the dataset for that measure.

\subsection{Computational results}

\paragraph*{Comparisons with max-normalization}

Each subplot of Figure \ref{Fig:histograms} shows four distributions comparing the PLR, ELRF, LRF, and RF metrics represented in blue, light orange, green, and red, respectively.

The ELRF, LRF, and RF distributions exhibit right-skewness, indicating that many data points cluster towards higher values. This skewness suggests a higher frequency of larger distances, a common trait among these metrics. Notably, the RF metric often shows smaller distances because it ignores label changes, whereas the ELRF and LRF metrics yield almost identical values, performing very similarly, as expected.

In contrast, the PLR distribution is centered around its mean, displaying a broader spread of measurements. This symmetric distribution indicates that the PLR metric has a greater variability in distance measurements, highlighting its sensitivity, that is, 
  a balanced penalization of all the elements of an evolutionary scenario. This contrasts with the more concentrated and nearly identical distributions of ELRF, LRF, and RF.

\begin{figure}[ht]
    \centering
    \includegraphics[width=0.9\textwidth]{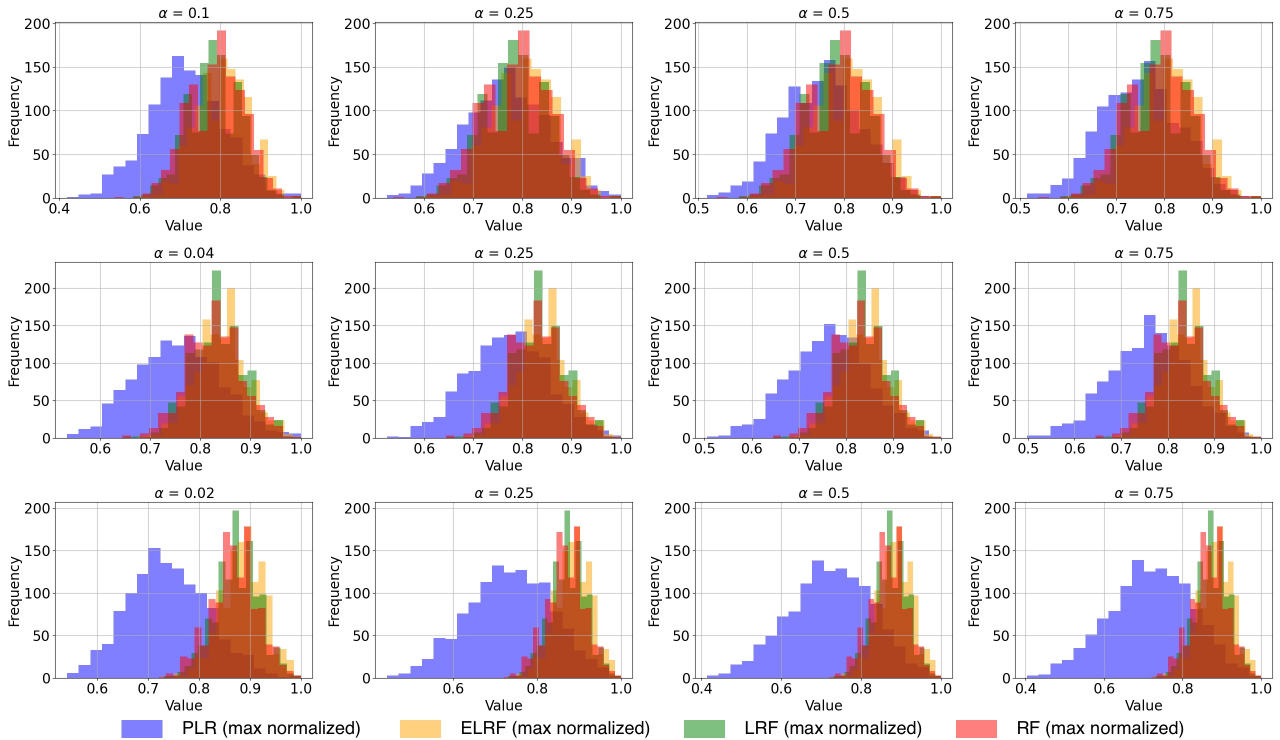}
    \caption{Distributions of the PLR, ELRF, LRF, and RF metrics for datasets $\Gamma_{10,20}$, $\Gamma_{25,10}$, and $\Gamma_{50,5}$, from top to bottom rows, respectively,  and alpha values from the set $\{\frac{1}{n}, 0.25, 0.5, 0.75\}$, with $n$ as number of species. Each row corresponds to a dataset, while each column represents a different value of $\alpha$. The $x$-axis represents max-normalized values ranging from $0$ to $1$, and the $y$-axis is the frequency of these values. The PLR measure in purple shows a centered and symmetric distribution with a broader spread. The ELRF, LRF, and RF metrics, shown in light orange, green, and red, respectively, exhibit right-skewed distributions towards the higher end of the scale.
}
    \label{Fig:histograms}
\end{figure}

\paragraph*{The theoretical diameter is hard to reach}

Figure \ref{fig:violin_plot} presents the distribution of the ELRF distance and the \metricshortname~distance for various values of the parameter $\alpha$. We omit the plots for LRF and RF distances since they closely resemble the ELRF distributions, as discussed in the previous section.

\begin{figure}[h]
  \centering
  \begin{tabular}{cc}
    \includegraphics[width=0.48\textwidth]{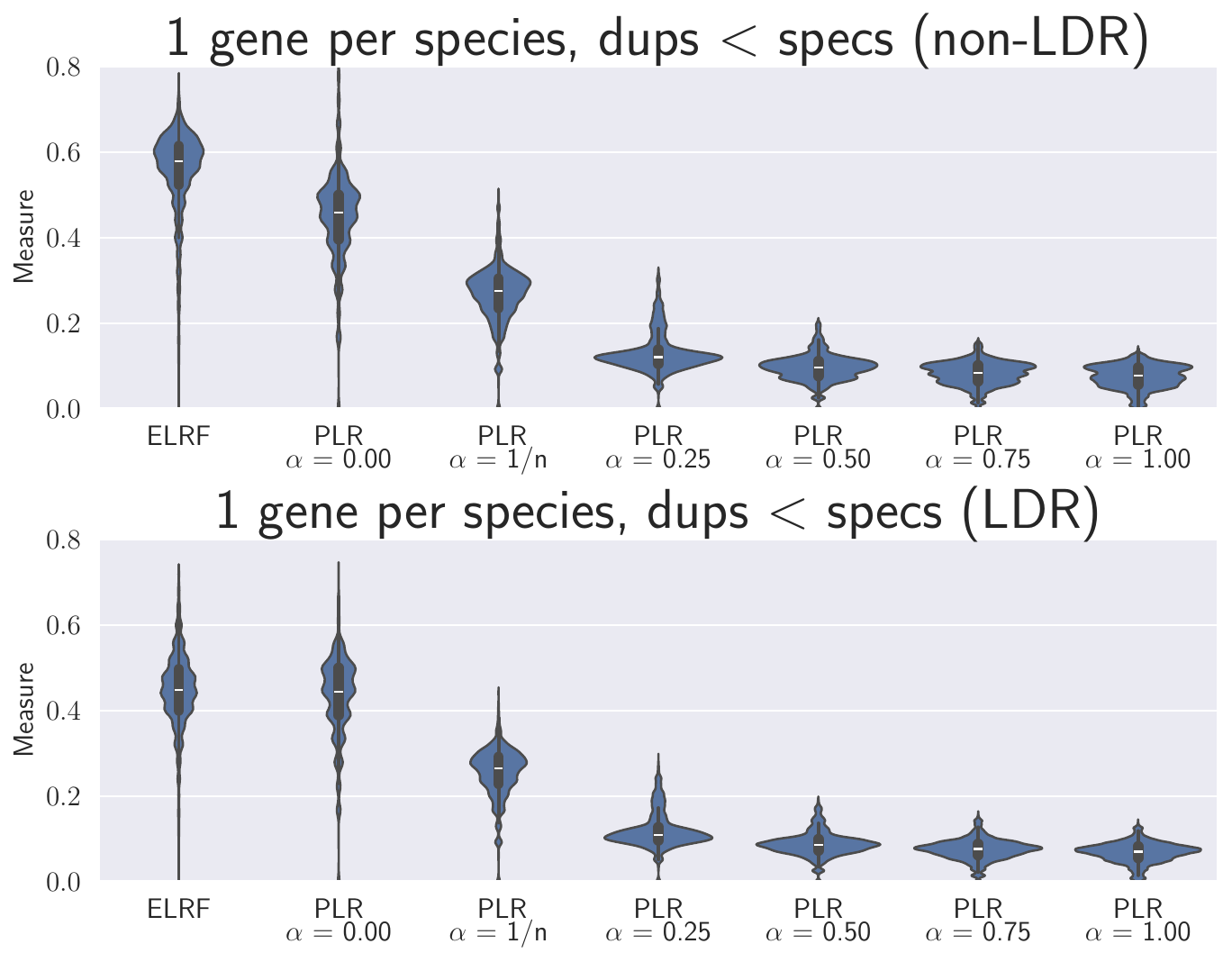} & \includegraphics[width=0.48\textwidth]{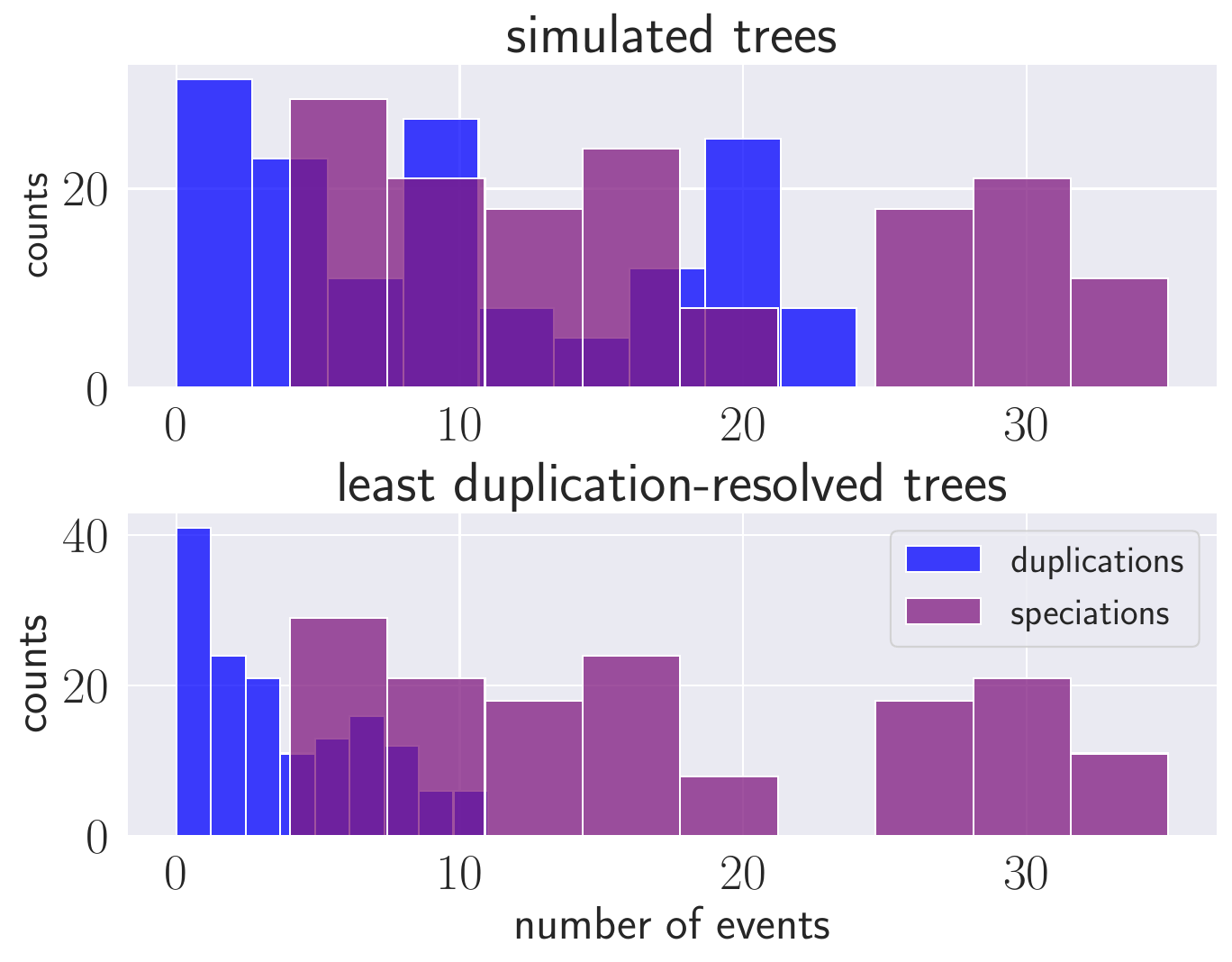}\\
    \includegraphics[width=0.48\textwidth]{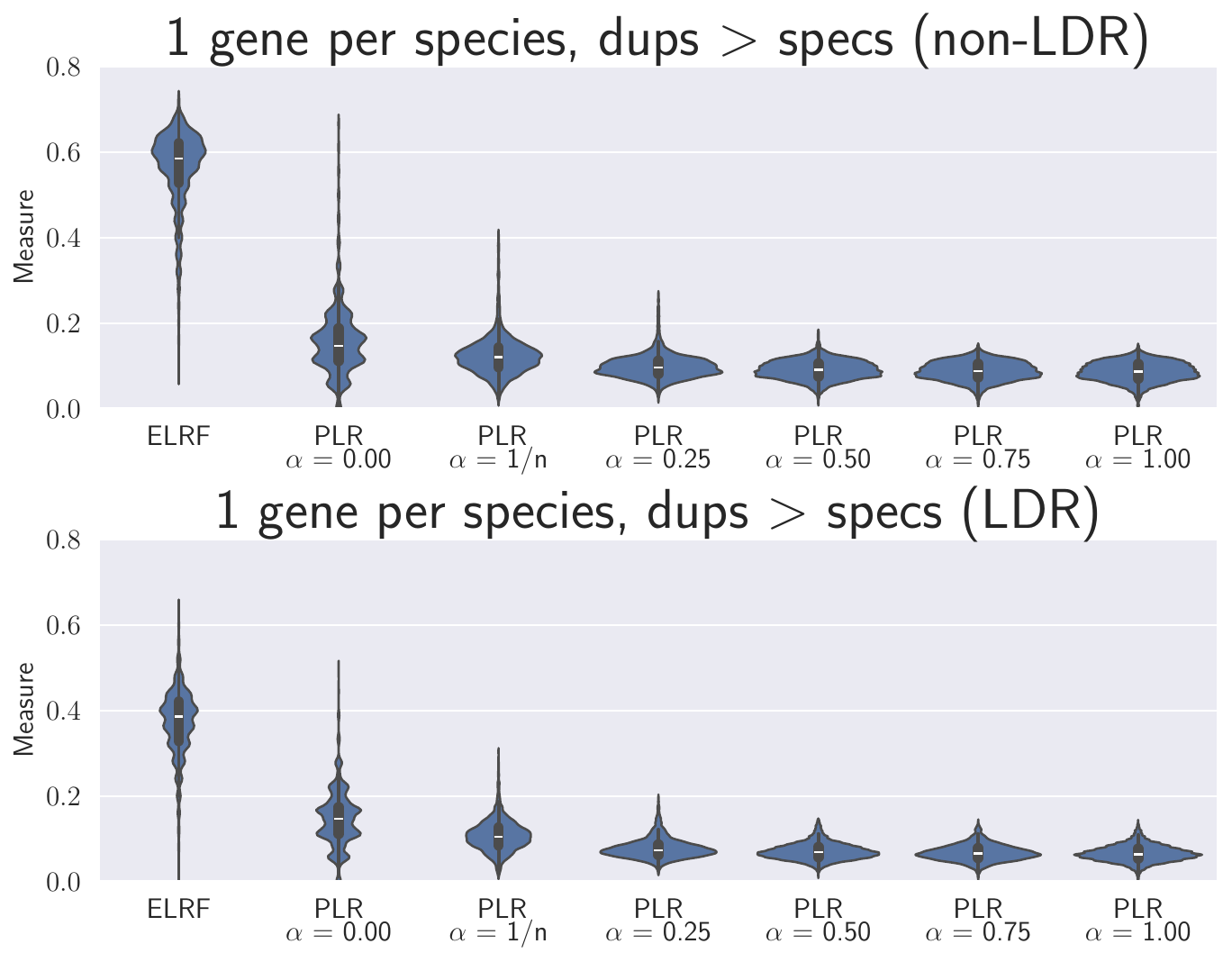} & \includegraphics[width=0.48\textwidth]{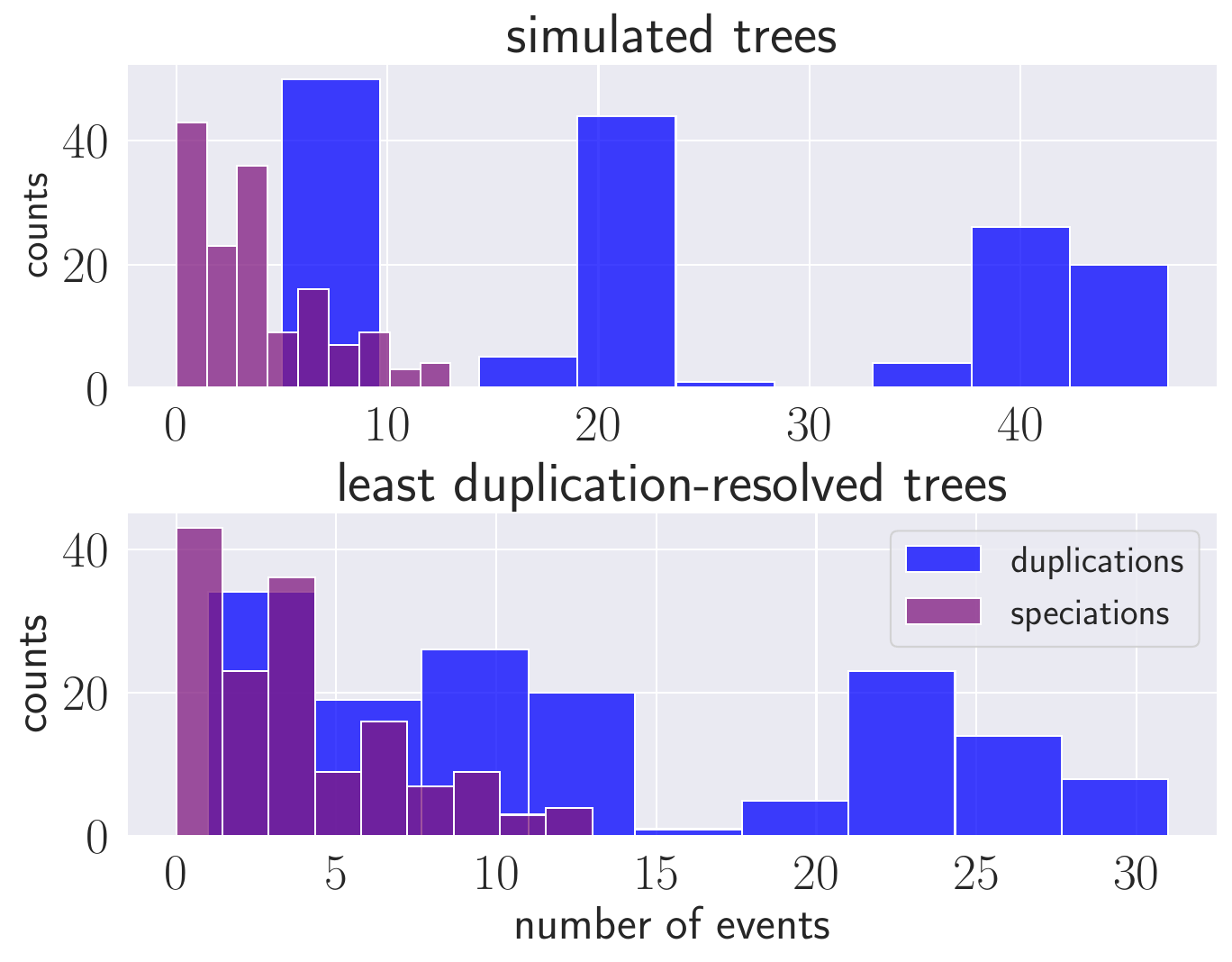}
  \end{tabular}
  \caption{Comparison of the distribution of ELRF and \metricshortname~measures with different values for the parameter $\alpha$, and different proportions of duplication/speciation events. The measures are shown for both the least duplication-resolved trees (LDR) and non-LDR. All the plots consider reconciliations with 10, 25, and 50 species. The parameter $\alpha=1/n$ aims to balance the linear-vs-quadratic components of the distance, where $n$ is the number of species.
Note that the biggest change in the distribution of the \metricshortname~measure happens for small values of $\alpha$.}
  \label{fig:violin_plot}
\end{figure}

The first two rows in Figure \ref{fig:violin_plot} compare trees with fewer duplications than speciations, while the subsequent rows involve trees with an equal or greater number of duplications compared to speciations. The PLR measure is normalized by the theoretical diameter introduced here, whereas the ELRF is normalized by its upper bound. Note that ELRF consistently has higher values than PLR and that these values are significantly far from the theoretical diameter.
The shape of the \metricshortname~distribution remains largely unchanged as $\alpha$ increases, likely due to the diminishing contribution of the linear component relative to the quadratic component as $\alpha$ grows.
On the right side of the figure, we observe the frequency of speciation and duplication events in our simulated reconciled trees, as well as their least duplication-resolved (LDR) counterparts. Notably, when there are more speciations than duplications, the PLR measure increases but still remains far from the theoretical diameter.

Figure \ref{fig:example-recons} illustrates important differences between the measures, since we can observe two different scenarios: 1) where ELRF is significantly smaller than \metricshortname, suggesting that reconciliations may be completely different even when gene tree topologies are similar; and 2) conversely, \metricshortname~may be significantly small when the ELRF is large, suggesting that different gene tree topologies could have similar reconciliations.

\begin{figure}[h]
    \centering
    \includegraphics[width=0.5\textwidth]{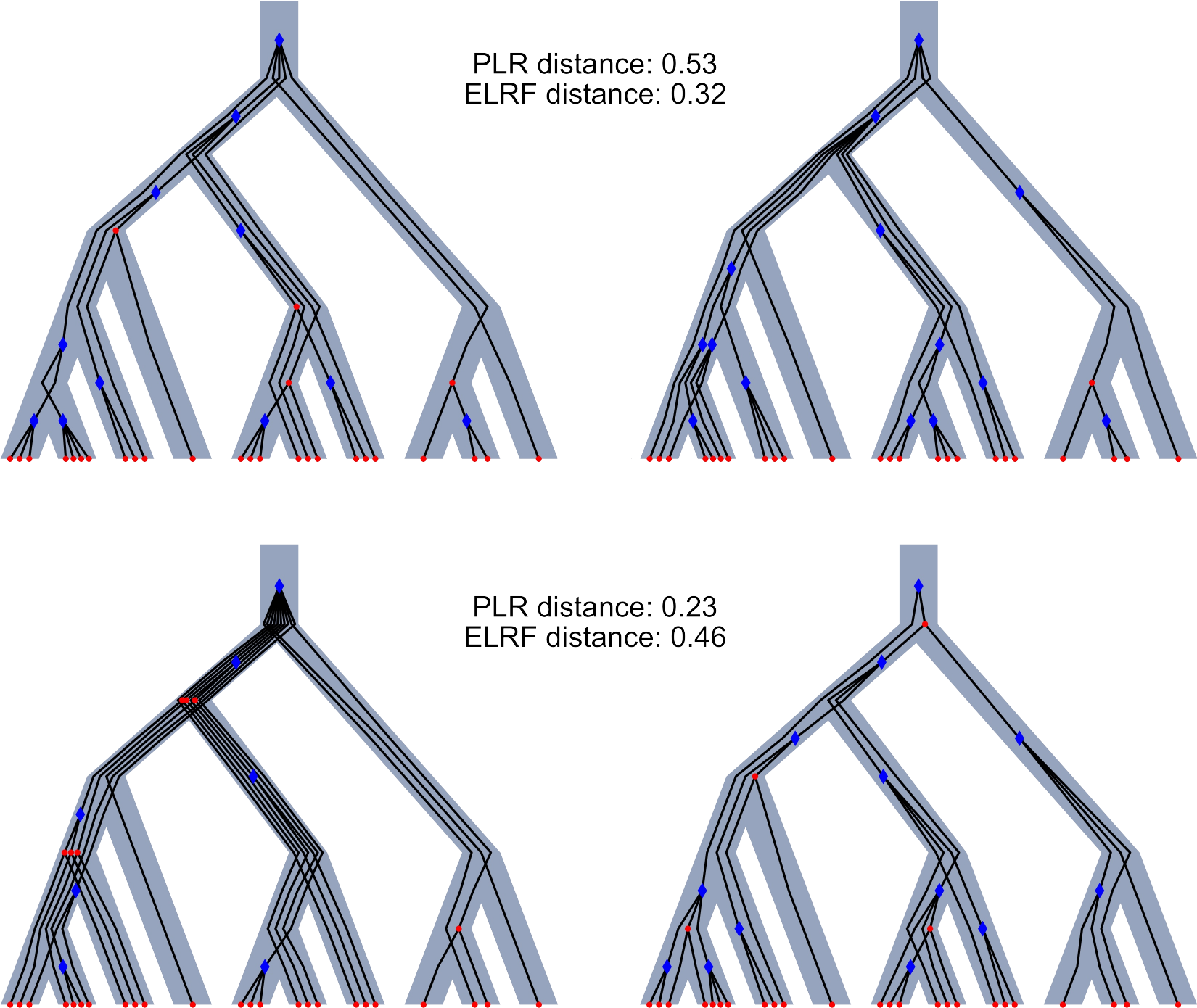}
    \caption{Examples of distance between reconciliations and gene trees, plotted using \texttt{REvolutionH-tl} \cite{revolutionhtl}. The reconciliations have 10 species and 24 genes, with $\alpha=1/10$. The upper row has a large \metricshortname~value but a small ELRF distance. In contrast, the bottom row shows trees when \metricshortname~is small even when ELRF is big.
    In this example, we set $\alpha=1/10$.}
    \label{fig:example-recons}
\end{figure}

\section{Discussion}

In this work, we have underscored the unique attributes of \metricshortname, a novel semi-metric designed for comparing reconciled gene trees within a fixed species tree framework. Unlike existing metrics such as RF, LRF, and ELRF, which primarily focus on tree topology, \metricshortname~incorporates all elements of an evolutionary scenario: a species tree, gene trees, speciation/duplication labeling and a mapping from gene trees to species tree. This broader scope provides a more holistic measure of dissimilarity between reconciled gene trees, offering researchers a nuanced understanding of evolutionary relationships.

One notable advantage of \metricshortname~is its flexibility, particularly regarding the parameter $\alpha$, which allows users to balance the quadratic and linear components of the distance according to their specific research context. This flexibility enhances the metric's applicability across diverse evolutionary scenarios, providing researchers with a customizable tool for reconciliation analysis. Additionally, our experiments reveal that \metricshortname~exhibits a symmetric and broadly spread distribution around its mean, indicating sensitivity to variations in reconciliations and finer granularity in distinguishing between different tree pairs. Despite its strengths, \metricshortname~does have some limitations. For instance, while the flexibility of $\alpha$ is advantageous, it also introduces a degree of subjectivity into the metric's application, as users must determine the appropriate value for their specific context. Moreover, our theoretical analysis highlights a large theoretical diameter for \metricshortname, which is seldom reached in practice. Tighter bounds are needed to improve practical applicability and interpretability. One of the key strengths of \metricshortname~is its computational efficiency, with an $O(n)$ time complexity. This efficiency is particularly beneficial for analyzing large datasets or trees, where computational resources and time are critical constraints. 

Looking ahead, future directions for \metricshortname~include refining the theoretical bounds of its diameter.
An important theoretical problem that remains open is determining whether \emph{binary} gene trees satisfy the triangle inequality. Additionally, developing metrics between gene trees with different leaf sets would significantly broaden its applicability. Incorporating alternative methods for matching ancestral genes, such as those proposed by Lin et al. \cite{lin2011metric}, or using asymmetric cluster affinity as suggested by Wagle \cite{Wagle2024}, could further enhance the metric's accuracy and relevance.

In conclusion, \metricshortname~represents a significant advancement in the comparison of reconciled gene trees, offering a detailed and flexible measure of dissimilarity. Its computational efficiency and comprehensive event consideration make it a valuable tool for evolutionary studies, with potential for further refinement and application in future research.

\section*{Acknowledgements}
The authors would like to thank the reviewers for their helpful comments.

\section*{Funding}
Alitzel L\'opez S\'anchez acknowledges financial support from the programme de bourses d'excellence en recherche from the University of Sherbrooke.

José Antonio {Ramírez-Rafael}  ackowledges financial support from the CONAHCYT scholarship, Mexico.

Manuel Lafond acknowledges financial support from the Natural Sciences and Engineering Research Council (NSERC) and the Fonds de Recherche du Québec Nature et technologies (FRQNT).

\bibliography{preprint_plr}

\end{document}